\documentclass{sig-alternate}

\usepackage{pgfplots}

\begin{document}

\conferenceinfo{SIGMETRICS'13,} {June 17--21, 2013, Pittsburgh, PA, USA.}
\CopyrightYear{2013}
\crdata{978-1-4503-1900-3/13/06}
\clubpenalty=10000
\widowpenalty = 10000


\newtheorem{mr}{Main result}
\newtheorem{theorem}{Theorem}[section]
\newtheorem{prop}[theorem]{Proposition}
\newtheorem{lemma}[theorem]{Lemma}
\newtheorem{corollary}[theorem]{Corollary}

\newcommand{\remark}[1]{\noindent \textbf{Remark.} #1}

\renewcommand{\topfraction}{0.85}
\renewcommand{\textfraction}{0.1}
\renewcommand{\floatpagefraction}{0.85}

\renewcommand{\a}{{\bf{a}}}
\newcommand{\A}{{\bf{A}}}
\newcommand{\q}{{\bf{q}}}
\newcommand{\Q}{{\bf{Q}}}
\renewcommand{\S}{{\bf{S}}}
\newcommand{\w}{{\bf{w}}}
\newcommand{\W}{{\bf{W}}}
\newcommand{\x}{{\bf{x}}}

\newcommand{\E}{\mathbb{E}}
\newcommand{\N}{\mathbb{N}}
\renewcommand{\P}{\mathbb{P}}
\newcommand{\R}{{\mathbb{R}}}
\newcommand{\Z}{{\mathbb{Z}}}

\newcommand{\Tcal}{{\mathcal{T}}}
\newcommand{\Rcal}{{\mathcal{R}}}

\newcommand{\drift}{\texttt{d}}

\newcommand{\indicator}[1]{{\rm I}_{\{#1\}}}

\newcommand{\Lone}[1]{\lVert #1 \rVert}
\newcommand{\Linfty}[1]{\lvert #1 \rvert}

\newcommand{\todo}[1]{\fbox{\parbox{.465\textwidth}{\begin{center} \textbf{#1} \end{center}}}}
\newcommand{\generalcomments}[1]{\fbox{\parbox{.465\textwidth}{\textbf{GENERAL COMMENTS. #1}}}}
\newcommand{\comments}[1]{\fbox{\parbox{.465\textwidth}{\textbf{#1}}}}

\title{Lingering Issues in Distributed Scheduling}

\numberofauthors{1}
\author{
\alignauthor
Florian Simatos, Niek Bouman, Sem Borst\titlenote{This work was supported by Microsoft Research through its PhD Scholarship Programme and a TOP grant from NWO.}\\
       \affaddr{Eindhoven University of Technology}\\
       \affaddr{Den Dolech 2}\\
       \affaddr{5612 AZ Eindhoven, The Netherlands}\\
       \email{\{f.simatos, n.bouman, s.c.borst\}@tue.nl}
}

\date{\today}

\maketitle

\begin{abstract}
Recent advances have resulted in queue-based algorithms for medium
access control which operate in a distributed fashion, and yet achieve
the optimal throughput performance of centralized scheduling algorithms.
However, fundamental performance bounds reveal that the ``cautious''
activation rules involved in establishing throughput optimality tend
to produce extremely large delays, typically growing exponentially
in $1/(1{-}\rho)$, with $\rho$ the load of the system,
in contrast to the usual linear growth.

Motivated by that issue, we explore to what extent more ``aggressive''
schemes can improve the delay performance.
Our main finding is that aggressive activation rules induce a lingering
effect, where individual nodes retain possession of a shared resource
for excessive lengths of time even while a majority of other nodes idle.
Using central limit theorem type arguments, we prove that the idleness
induced by the lingering effect may cause the delays to grow with
$1/(1{-}\rho)$ at a quadratic rate.
To the best of our knowledge, these are the first mathematical results
illuminating the lingering effect and quantifying the performance impact.

In addition extensive simulation experiments are conducted to illustrate
and validate the various analytical results.
\end{abstract}



\section{Introduction}

As networks continue to grow in size and complexity,
they increasingly rely on scheduling algorithms for efficient allocation
of shared resources and arbitration between users.
As a result, the design and analysis of scheduling algorithms
for complex network scenarios has attracted significant attention
over the last several years.

One of the centerpieces in the scheduling literature is the celebrated
MaxWeight algorithm as proposed in the seminal work \cite{TE92,TE93}.
The MaxWeight algorithm provides throughput optimality and maximum
queue stability in a variety of scenarios, and has emerged
as a powerful paradigm in cross-layer control and resource allocation
problems~\cite{GNT06}.

While not strictly optimal in terms of delay performance, MaxWeight
algorithms do achieve so-called equivalent workload minimization
and offer favorable scaling characteristics in heavy traffic conditions
\cite{Shah12:1,Stolyar04}.
As a further key appealing feature, MaxWeight algorithms only need
information on the queue lengths and instantaneous service rates,
and do not rely on any explicit knowledge of the underlying system
parameters.

On the downside, solving the maximum-weight problem tends to be
challenging and potentially NP-hard. This is exacerbated
in a network setting, where a centralized control entity may be
lacking or require global state information, creating a substantial
communication overhead in addition to the computational burden.
This concern is especially pertinent as the maximum-weight problem
needs to be solved at a high pace, commensurate with the fast time
scale on which scheduling algorithms typically need to operate.

This issue has provided a strong impetus for devising algorithms that
entail \emph{lower computational complexity and communication overhead},
but retain the \emph{maximum stability and throughput guarantees}
of the MaxWeight algorithm.
Various approaches in that direction were proposed in
\cite{CS06,MSZ06,SMS06,SSM07,Tassiulas98,WS06}.
An exciting breakthrough in this quest was recently achieved in the
design of random back-off schemes for wireless medium access control
that seem to offer the best of both worlds.
These schemes operate in a \emph{distributed fashion, requiring no
centralized control entity or global state information}, and yet,
remarkably, provide the \emph{capability of achieving
throughput optimality and maximum stability}.
More specifically, clever algorithms have been developed for finding
the back-off rates that yield any given target throughput vector
in the capacity region \cite{JW08,ME08}.
In the same spirit, powerful algorithms have been devised for adapting
the back-off rates based on queue length information, and been shown
to guarantee maximum stability \cite{Jiang10:0,RSS09a,Shah12:0}.

While the maximum-stability guarantees for the above-mentioned
algorithms have strong appeal, they do not extend to performance metrics
such as expected queue lengths or delays.
In fact, fundamental performance bounds~\cite{Bouman11:1} indicate that
the ``cautious'' back-off functions involved in establishing maximum
stability tend to produce extremely large delays, typically growing
\emph{exponentially} in $1/(1 {-} \rho)$, with $\rho$ the load of the system,
in contrast to the usual \emph{linear} growth.
More specifically, the bounds show that the expected queue lengths grow
as $\psi^{- 1}(1 {-} \rho)$ as $\rho \uparrow 1$.
Here $\psi^{- 1}$ represents the inverse of the (decreasing)
function $\psi$, specifying the probability of a node entering
a back-off as a function of its current queue length.
The bounds may be explained by noting that the queue lengths govern the
fraction of back-off time through the function $\psi$.
Since the fraction of back-off time cannot exceed the surplus capacity
in order for the system to be stable, however, it is ultimately the
amount of surplus $1 {-} \rho$ that dictates the queue lengths through
the function $\psi^{- 1}$.
We note that maximum stability has been established under the
condition that the function $\psi(a)$ decays (no faster than)
\emph{inverse-logarithmically} as $a \to \infty$, i.e., $\psi(a) \sim 1/\log a$.
This entails that $\psi^{- 1}(s)$ grows (no slower than)
\emph{exponentially} in $1 / s$ as $s \downarrow 0$,
yielding the stated exponential growth of $\psi^{- 1}(1 {-} \rho)$
in $1/(1 {-} \rho)$ as $\rho \uparrow 1$.

The above lower bounds suggest that the delay performance may be
improved when the function $\psi$ decays faster, e.g.,
\emph{inverse-polynomially}: $\psi(a) \sim a^{-\beta}$, with $\beta > 0$,
so that $\psi^{- 1}(s) \sim s^{- 1 / \beta}$ as $s \downarrow 0$.
The larger the value of the exponent~$\beta$, the slower the growth of
$\psi^{- 1}(1 {-} \rho) \sim (1 {-} \rho)^{- 1 / \beta}$ as $\rho \uparrow 1$.
In particular, it might seem plausible that for $\beta \geq 1$,
the expected queue lengths will only exhibit the usual linear growth
in $1 / (1 {-} \rho)$ as $\rho \uparrow 1$.
Note that a larger value of~$\beta$ means that a node is more
``aggressive'', in the sense that it is less likely to enter a back-off
and more inclined to hold on to the medium, and hence the
coefficient~$\beta$ will be referred to as the aggressiveness parameter.
It is worth observing that maximum stability for the above back-off
functions is not guaranteed by existing results,
which do not apply for any $\beta > 0$.
In fact, for $\beta > 1$, maximum stability has been shown not to hold
in certain topologies~\cite{GBW12}. \\

In the present paper we aim to gain fundamental insight whether
a larger aggressiveness parameter can improve the delay performance.
Our main finding is that for large values of $\beta$,
a \emph{lingering effect} can cause the mean stationary delay
to increase in heavy traffic as $1/(1{-}\rho)^2$,
and we focus on the simplest topology where this effect occurs.
This topology, described in later sections, may seem at first sight
rather restrictive in view of recent results
\cite{Jiang10:0,RSS09a,Shah12:0} that apply to general topologies.
We believe however that our results give insight into more general
situations and discuss this in Section~\ref{sec:insight}.

The remainder of the paper is organized as follows.
We start in Section~\ref{sec:example} by discussing an example that
will explain the model that we consider as well as the lingering effect.
In Section~\ref{mode} we present a detailed model description.
In Section~\ref{ovmare} we provide an overview of the results
and discuss the main performance implications.
We conduct comprehensive simulation experiments to illustrate the
various heavy traffic results in Section~\ref{ec:simula}.
In Section~\ref{asan} we present a detailed asymptotic analysis
and proof arguments, and we conclude in Section~\ref{sec:insight}
by discussing the broader implications of our results.

\section{Illustrative example} \label{sec:example}


Consider a network consisting of four queues which are split into two
groups, say groups~1 and~2, in such a way that if any queue of group~1
is transmitting a packet, no queue of group~2 may transmit.
A group is said to be \emph{active} if one of its queues is transmitting a packet, and a queue is said to be active if it belongs to the active group. The other group and queues are said to be \emph{inactive}. 

Active queues implement the following algorithm:
after each transmission, each active queue flips a coin
and \emph{advertizes a release} with probability $(1+a)^{-\beta}$, with
$a$ the number of packets that this queue has to transmit and $\beta > 0$.
If the two active queues advertize a release simultaneously,
then active queues become inactive and vice-versa:
such a time is called a \emph{switching time}.
This simple distributed algorithm gives rise to dynamics as illustrated
in Figure~\ref{fig:realtimezoom}, where the system is considered
over three consecutive switching times $t_1$, $t_2$ and $t_3$.
Between switching times, the backlog of active queues is drained
while packets accumulate at inactive queues.
The dynamics shown in Figure~\ref{fig:realtimezoom} are representative
of the case $\beta > 1$ where a switch does not occur until both
active queues have emptied as will be established later.

\begin{figure}[htbp]
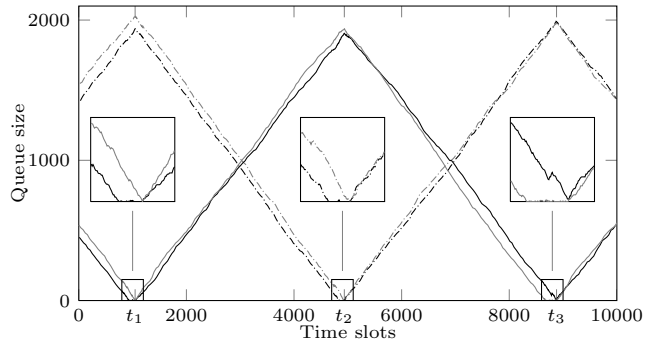

\vspace{2mm}
\begin{tikzpicture} [every pin/.style={fill=none}]
\begin{axis}[
font=\scriptsize,
width=\columnwidth+0.3cm, height=5.5cm,
    xmin=0, xmax=10000,
    ymin=0, ymax=2100,
    xtick={0,1047,2000,4000,4935,6000,8000,8879,10000},
    xticklabels={$0$,$t_1$,$2000$,$4000$,$t_2$,$6000$,$8000$,$t_3$,$10000$},
    every axis x label/.style=
        {at={(ticklabel cs:0.5)},anchor=center},
    ytick={0,1000,2000},
    yticklabels={$0$,$1000$,$2000$},
    every axis y label/.style=
        {at={(ticklabel cs:0.5)},rotate=90,anchor=center},
    scaled ticks=false,
	xlabel={Time slots},
	ylabel={Queue size}]
	\input{fig/thinned}

\coordinate (spypoint) at (axis cs:1000,150);
\coordinate (spypoint2) at (axis cs:4900,150);
\coordinate (spypoint3) at (axis cs:8800,150);
\draw [black] (axis cs:800,0) rectangle (axis cs:1200,150);
\draw [black] (axis cs:4700,0) rectangle (axis cs:5100,150);
\draw [black] (axis cs:8600,0) rectangle (axis cs:9000,150);
\end{axis}
\node[pin={[pin distance=0.8cm]90:{%
    \begin{tikzpicture}[baseline,trim axis left,trim axis right]
    \begin{axis}[
    width=2.7cm, height=2.7cm,
    xmin=800, xmax=1200,
    ymin=0, ymax=150,
    ticks=none]
	\input{fig/spy1}
    \end{axis}
    \end{tikzpicture}%
}}] at (spypoint) {};

\node[pin={[pin distance=0.8cm]90:{%
    \begin{tikzpicture}[baseline,trim axis left,trim axis right]
    \begin{axis}[
    width=2.7cm, height=2.7cm,
    xmin=4700, xmax=5100,
    ymin=0, ymax=150,
    ticks=none]
    \input{fig/spy2}
    \end{axis}
    \end{tikzpicture}%
}}] at (spypoint2) {};

\node[pin={[pin distance=0.8cm]90:{%
    \begin{tikzpicture}[baseline,trim axis left,trim axis right]
    \begin{axis}[
    width=2.7cm, height=2.7cm,
    xmin=8600, xmax=9000,
    ymin=0, ymax=150,
    ticks=none]
    \input{fig/spy3}
    \end{axis}
    \end{tikzpicture}%
}}] at (spypoint3) {};

\end{tikzpicture}
\vspace{-7mm}
\caption{A typical sample path in the case $\beta > 1$.
The three boxes zoom in to show the lingering effect.
One queue hovers around zero while the other queue is yet to empty,
resulting in an inefficient use of the resource.}
\label{fig:realtimezoom}
\end{figure}

Let us now give a flavor of the lingering effect. Imagine that the two active queues start with initial queue lengths
of the same order, say~$Q$. As just mentioned, queues retain the shared resource
until the time $T^*$ at which the last active queue has emptied, thus preventing other queues from activating until this time.
The law of large numbers suggests that $T^*$ is of order~$Q$ (i.e., active queues are drained linearly as in Figure~\ref{fig:realtimezoom}),
but the central limit theorem suggests that up to time~$T^*$, the first active queue to have emptied, while waiting for the other queue to empty,
will be empty of the order of $\sqrt Q$ units of time.
This lingering effect is illustrated in Figure~\ref{fig:realtimezoom},
which will explained in greater detail in Section~\ref{ec:simula}.

This $(1/\sqrt Q)$-fraction of the time the shared resource is used
inefficiently may at first sight seem negligible,
and indeed queues seem to empty at the same time on the coarse time
scale of Figure~\ref{fig:realtimezoom}.
However, we will actually establish that it has a significant impact
in heavy traffic, causing queue lengths to grow at a rate $1/(1-\rho)^2$
as $\rho \uparrow 1$, instead of the optimal $1/(1-\rho)$.



\section{Model description}
\label{mode}

\subsection{Informal description}
\label{sub:informal-description}

Let us now give a more precise definition of our model.
As mentioned in the introduction, in order to analyze the lingering
effect in the simplest possible setting, we focus on a symmetric
system consisting of two groups of $R \geq 2$ queues.
At any given point in time, one of the groups is \emph{active}
while the other is \emph{inactive}.

Time is slotted and inactive queues have simple dynamics,
driven by independent and identically distributed numbers of packet
arrivals in each slot, so that they each simply grow according
to random walks with step distribution denoted by~$\xi$.
During each time slot, active queues increase by independent amounts
distributed as~$\xi$ as well, but, if at least one packet is present at
the start of the slot, an active queue also flushes exactly one packet.

Moreover, at the end of each time slot, each active queue tosses a coin
and advertizes a \emph{momentary release} with probability $\psi(a)$,
with $a$ the number of packets in the queue at the end of the time slot.
This (momentary) release gives inactive queues an opportunity to become
active: if all the active queues \emph{simultaneously} advertize a release,
then inactive queues become active and active queues become inactive.
Such a time is called a \emph{switching time}.
We will assume in the sequel that $\psi(a) = (1+a)^{-\beta}$ for some
parameter $\beta > 0$, called the aggressiveness parameter.
In particular, $\psi(a) \to 0$ as $a \to +\infty$, and so active queues
are less likely to advertize a release when they are highly loaded;
this mechanism thus gives priority to highly loaded queues
in a distributed fashion.
There is a cost associated with advertizing a release:
each time an active queue advertizes a release and is not empty,
it incurs an additional increase distributed according to some random
variable~$\zeta$.

The above-described model qualitatively resembles the ca\-nonical
models for queue-based medium access control mechanisms.
The main difference with continuous-time models in the literature is
that in our model, back-off periods are infinitesimally short
(hence, the releases are qualified as momentary),
and the jump size~$\zeta$ represents the number of packets that would
have arrived during a non-zero back-off period.
The key advantage of our model is that it avoids the complication
of requiring back-off periods to overlap to define a switching time.
Moreover, simulation experiments show that our main results extend
to models with non-infinitesimal back-off periods.
We did not include these simulations to comply with page limitations.
\\

\begin{remark}
	%
	%
	The system may be interpreted
	as a set of $R$ single-server two-queue polling systems,
	where servers are only allowed to switch between queues
	in a synchronized fashion.
	Such models with simultaneous service of several queues
	and synchronized switches are natural in applications (e.g., traffic lights
	at intersections), but appear not to have been considered.
\end{remark}

\subsection{Parameters}

The model described in the previous subsection is defined through
four parameters:
the number $R \geq 2$ of queues in each group,
two integer-valued random variables~$\xi, \zeta \in \N = \{0,1,\ldots\}$
and a real number $\beta \in (0,\infty]$, which defines the $[0,1]$-valued
sequence $(\psi(a), a \geq 0)$ via $\psi(a) = (1+a)^{-\beta}$ for
$a \geq 0$, to be understood for $\beta = +\infty$ as $\psi(0) = 1$
and $\psi(a) = 0$ for $a > 0$.
Note that only the asymptotic behavior of $\psi(\cdot)$ matters,
and our results could easily be extended to any $\psi(\cdot)$
with $a^\beta \psi(a) \to \ell \in (0,\infty)$ as $a \to +\infty$. \\
We assume that $\xi$ and $\zeta$ have finite means, respectively
$\E\xi = \rho/2$ and $\E\zeta = z$, and that $\xi$ has finite variance
denoted by $v = \E[(\xi-\rho/2)^2]$.
It will be argued that the system is stable if and only if $\rho < 1$
and so we will refer to $\rho$ as the load of the system (note that by symmetry, each queue is active half the time, which explains the factor $2$ in the definition $\rho = 2\E\xi$ of $\rho$).

\subsection{Formal description}
\label{sub:formal-description}

Because of the symmetry of the system, we do not need to label the
queues individually, but only need to keep track of the state of
active and inactive queues.
We will consider the system embedded at switching times, and define
$Q^a_r(k)$ and $Q^i_r(k)$ as the numbers of packets in the $r$th active
and inactive queue, respectively, just after the $k$th switch occurred.
We will be interested in the Markov chain $(\Q(k), k \geq 0)$
which we also write as $\Q = (\Q^a, \Q^i)$
with $\Q^a = (Q_r^a, r \in \Rcal)$, $\Q^i = (Q_r^i, r \in \Rcal)$, $\Rcal = \{1,\ldots, R\}$ and we reserve in the sequel bold
notation for vectors (of functions or numbers).

As informally described in Section~\ref{sub:informal-description},
the dynamics of $\Q$ in between two switching times are governed by two
$R$-dimensional processes $\S = (S_r, r \in \Rcal)$
and $\A = (A_r, r \in \Rcal)$: $\S$ gives the increments of the
inactive queues, while $\A$ gives the state of the active queues.
The dynamics are as follows:
\begin{itemize}
\item the $2R$ processes $A_r, S_r$ are independent;
\item for each $r \in \Rcal$, $(S_r(k), k \geq 0)$ is a random walk
with step distribution~$\xi$, started at~0;
\item for each $r \in \Rcal$, $(A_r(k), k \geq 0)$ is
a space-inho\-mogeneous random walk with the following dynamics:
for any $a \in \N$ and any function $f: \N \to [0,\infty)$, we have
	\begin{multline} \label{eq:A}
		\E\left[f(A_r(1)) \mid A_r(0) = a\right]\\
		= \E\left[f\left(Y(a) + \zeta \indicator{Y(a) > 0, U < \psi(Y(a))} \right) \right],
	\end{multline}
	where $Y(a) = a + \xi - \indicator{a > 0}$ is the number of packets at the end of the time slot and $U$ is uniformly distributed in $[0,1]$, and $U$, $\xi$ and $\zeta$ are independent.
\end{itemize}
The equation~\eqref{eq:A} describes the dynamics of an active queue
and can be interpreted as follows. At each time slot, an active queue increases by~$\xi$
and, if not empty at the beginning of the time slot, flushes a packet, which brings the queue to state $Y(a)$. If $U < \psi(Y(a))$, we say that
the active queue advertizes a release, which thus happens with the
(conditional) probability $\psi(Y(a))$ as described in
Section~\ref{sub:informal-description}.
If at the end of the time slot the queue is not empty and it advertizes
a release, i.e., $Y(a) > 0$ and $U < \psi(Y(a))$,
then the active queue also incurs an additional increase by~$\zeta$.

To define the $2R$-dimensional process $\Q = (\Q^a, \Q^i)$ from $\S$
and $\A$, it remains to adopt notation for the switching time,
which we denote by~$T^*$.
Thus $T^*$ is the first time at which all active queues advertize
a release at the same time.
Note that $T^*$ and $\S$ are independent.
With these definitions, the dynamics of $\Q$ as informally described
in Section~\ref{sub:informal-description} obey the following equation:
for any $\q = (\q^a, \q^i) \in \N^R \times \N^R$ and any function
$f: \N^{2R} \to [0,\infty)$,
\begin{multline} \label{eq:Q}
	\E\left[f(\Q(1)) \mid \Q(0) = \q\right]\\
	= \E\left[f(\q^i + \S(T^*), \A(T^*)) \mid \A(0) = \q^a\right].
\end{multline}

The special case $\beta = +\infty$ will be of particular importance.
Indeed, we will show that it is representative of the system's behavior
in the range $\beta > 1$.
The case $\beta = +\infty$ has been studied in~\cite{Feuillet10:0},
which refer to it as the random capture algorithm.
When $\beta = +\infty$, active queues only advertize a release when
they are empty (in which case there is no additional term $\zeta$), and so $\A(T^*) = {\bf{0}}$ and $T^* = \inf \left\{ k \geq 0: \A(k) = {\bf{0}} \right\}$.

\subsection{Additional notation}

In the remainder of the paper, and similarly as we have just done in~\eqref{eq:Q}, we will use the common symbol $\E$ to denote expectation with respect to the laws of $\Q$ and $(\A, \S)$. Initial conditions will be denoted by a subscript, and it should always be clear from the context whether we consider initial conditions of $\Q$, $\A$ or some $A_r$ (remember that $\S(0)$ is always equal to ${\bf{0}}$). For instance,~\eqref{eq:A} and~\eqref{eq:Q} can be rewritten as follows:
\[ \E_a\left[f(A_r(1)) \right] = \E\left[f\left(Y(a) +
\zeta \indicator{Y(a) > 0, U < \psi(Y(a))} \right) \right], \]
and
\[ \E_{\q}\left[f(\Q(1)) \right] = \E_{\q^a}\left[f(\q^i + \S(T^*), \A(T^*)) \right]. \]
The probability distributions corresponding to these various
expectations are written as $\P_a$, $\P$, $\P_\q$ and $\P_{\q^a}$. We also define $\P_\infty$ with corresponding expectation $\E_\infty$ as the laws of $\Q$ and $(\A, \S)$ started in the stationary distribution of $\Q$, provided $\Q$ is positive recurrent.

When $\beta = +\infty$, we see based on~\eqref{eq:Q} that $\Q^i(k) = {\bf{0}}$, except maybe for $k = 0$.
In particular, when $\Q$ is positive recurrent, $\Q^i(0)$ is $\P_\infty$-almost surely equal to $\bf{0}$ and~\eqref{eq:Q} therefore becomes
\begin{equation} \label{eq:stationarity-RC}
	\E_\infty\left[f(\Q^a(0))\right] = \E_\infty\left[f(\S(T^*))\right] \quad (\beta = +\infty).
\end{equation}
Further, in the remainder of the paper we let
\[ \tau_r = \inf\{ k \geq 0: A_r(k) = 0 \}, \hspace*{.4in} \ \tau_{\max} = \max_{r \in \Rcal} \tau_r, \]
and we define the $\tau_{(r)}$'s as the order statistics of the $\tau_r$'s, i.e., $\tau_{(1)} \leq \cdots \leq \tau_{(R)}$ and $\{ \tau_{(r)} \} = \{ \tau_{r} \}$. Let finally $\Linfty{\cdot}$ be the $L_\infty$ norm and $\Lone{\cdot}$ be the $L_1$ norm, i.e., if $J \geq 1$ and $\x \in \R^J$ then $\Linfty{\x} = \max_j \Linfty{x_j}$ (which is just the absolute value for $J = 1$) and $\Lone{\x} = \Linfty{x_1} + \cdots + \Linfty{x_J}$.

\section{Overview of main results}
\label{ovmare}

\subsection{Two regimes}

As stated in the introduction, we aim to gain fundamental insight
in the impact of the function $\psi(\cdot)$,
through the aggressiveness parameter~$\beta$,
on the system performance in terms of expected queue lengths and delays.
We will demonstrate, based on a combination of heuristic arguments,
simulation experiments and theoretical results, that the system's
behavior changes as $\beta$ increases from~0 to~$\infty$.
In the small $\beta$ regime, typically $\beta < 1/2$, the system is
characterized by a mean-reverting effect that induces (stationary)
queue lengths of the order of $1/(1-\rho)^{1/\beta}$.
This regime was already hinted upon in~\cite{Bouman11:1},
where a corresponding lower bound was rigorously proved in a general setting.
In the present paper we go one step further, by providing a heuristic
argument in Section~\ref{sub:small-beta}, explaining why this lower
bound is sharp for small~$\beta$, and corroborating this by extensive
simulation results in Section~\ref{sub:simu-beta}.
However, when $\beta$ increases, this mean-reverting effect vanishes,
which causes the lower bound to become loose.
For $\beta > 1$, we show that the system's performance is dominated
by another phenomenon, which we call a lingering effect.
The investigation of this latter effect constitutes the main
contribution of the paper.
We provide a heuristic explanation of this effect
in Section~\ref{sub:lingering-effect},
examine it via simulation experiments in Section~\ref{sub:simu-beta},
and prove various theoretical results in Section~\ref{asan}.

\subsection{Small $\beta$: a mean-reverting effect}
\label{sub:small-beta}

Consider for a moment a given queue and let $\Delta(a)$ be the mean increase of this queue starting from level $a$. Because of symmetry, the queue will be active half the time and inactive the other half of the time, so that
\[ \Delta(a) \approx \frac{1}{2} \E_a \left[ \text{increase} \mid \text{inactive} \right] + \frac{1}{2} \E_a \left[ \text{increase} \mid \text{active} \right]. \]
From the dynamics described in Section~\ref{sub:formal-description}
we may deduce for large $a$ (approximating $Y(a) \approx a$
and neglecting the possibility of $Y(a) = 0$) that
$2\Delta(a) \approx \E\xi + \E(\xi - 1 + \zeta \indicator{U < \psi(a)})$,
which leads to \[ 2\Delta(a) \approx \rho - 1 + z \psi(a). \]
This last approximation points to a mean-reverting effect:
the average drift is negative for $a > a^*$ and positive for $a < a^*$,
where $a^*$ is determined by the equation $\rho - 1 + z \psi(a^*) = 0$,
which gives, since $\psi(a) = (1+a)^{-\beta}$,
$a^* \approx 1/(1-\rho)^{1/\beta}$.
Using the convexity of $\psi(\cdot)$ and Jensen's inequality,
it is not difficult to convert the above back-of-the-envelope
computations into a rigorous proof and show that $1/(1-\rho)^{1/\beta}$
is a lower bound for the stationary mean number of packets in the system,
as established in~\cite{Bouman11:1} in a continuous-time setting.

However, our explanation of this lower bound through this
mean-reverting effect goes one step further, and shows that this lower
bound should be sharp for small values of~$\beta$.
Indeed, $\Delta(a)$ is the mean drift obtained by averaging over the
active and inactive states.
Roughly speaking, this quantity describes the drift experienced
by a queue if the queue rapidly changes states, i.e., if the time
needed for the queue to switch state is negligible compared to~$n$,
with $n$ the length of the queue.
Otherwise, the queue stays in the active or inactive state for a long period of time, during which it experiences (almost) constant drift, either positive or negative.
This explanation bears some similarity with a random walk in
a dynamic random environment, where the random walk essentially sees
a constant environment if the random environment mixes rapidly.

To conclude our arguments, it remains to note that, at least intuitively,
queues will switch more rapidly for smaller values of~$\beta$.
We will provide simulation experiments showing that the stationary
mean number of packets grows as $1/(1-\rho)^{1/\beta}$ for small
values of~$\beta$.
As $\beta$ increases however, this mean-reverting effect vanishes.
The main contribution of this paper is to reveal and quantify
a previously unknown effect, called the lingering effect,
which explains the system's behavior in this case.

\subsection{Large $\beta$: the effect of lingering}
\label{sub:lingering-effect}

As described in the introduction, when the function $\psi(\cdot)$
decays sufficiently fast, once a queue gains possession of the resource,
it holds onto it, even when some or all of the other queues in the
same group are empty, and it would be more efficient for the queues
in the other group to receive the resource.
This causes a lingering effect as illustrated in
Figure~\ref{fig:realtimezoom} for a scenario with $R = 2$, $\beta = 2$.
It may appear that the two queues in the same group drain around the
same time (as can indeed be shown to be the case on a ``fluid scale'').
When we zoom in, however, we see that there is actually a time period
where one of the queues is already empty, while the other one clings to the
resource and prevents the two queues in the other group from activating.

Let us now provide a somewhat more technical description of this
phenomenon in the case $\beta = +\infty$, where it is easiest to see
thanks to the simplification of the dynamics~\eqref{eq:stationarity-RC}.
Moreover, we will argue that the case $\beta \in (1, +\infty)$ is
essentially a perturbation of the case $\beta = +\infty$.

By applying~\eqref{eq:stationarity-RC} to $f = \Lone{\cdot}$, we obtain
$\E_\infty \left( \Lone{\Q^a(0)} \right) = \E_\infty \left( \Lone{\S(T^*)} \right)$
and since $\Lone{\S(\cdot)}$ is a random walk with drift $R \E\xi$
independent of $T^*$ and $\A(0) = \Q^a(0)$ (by definition), we obtain
$\E_\infty \left( \Lone{\A(0)} \right) = R \E(\xi) \E_\infty\left( T^* \right)$
and so by symmetry,
\begin{equation} \label{eq:A_1}
	\E_\infty \left( A_1(0) \right) = \E(\xi) \E_\infty\left( T^* \right).
\end{equation}
The goal is now to relate $\E_\infty(T^*)$ to $\E_\infty(A_1(0))$. Remember that $T^* = \inf\{ k: \A(k) = {\bf{0}} \}$ when $\beta = +\infty$.
It is not difficult to show that $T^* \approx \tau_{\max}$,
essentially because once all queues have hit~0, it is only a matter
of constant time for all queues to be empty simultaneously
(this will be justified in Lemma~\ref{lemma:T^*}).
Then, the central limit theorem shows that
$\tau_r \approx A_r(0) / (1-\E\xi) + A_r(0)^{1/2}$
(where we neglect multiplicative constants, possibly random, appearing in front
of first- or second-order terms and that do not influence the order of magnitude of the final result), which leads to the approximation
$\tau_{\max} \approx |\A(0)| / (1-\E\xi) + |\A(0)|^{1/2}$.
Since under $\P_\infty$ queues are symmetric, we have $|\A(0)| \approx A_1(0) + A_1(0)^{1/2}$ which finally leads to $T^* \approx A_1(0) / (1-\E\xi) + A_1(0)^{1/2}$, i.e.,
\[ \E_\infty \left( A_1(0) \right) \approx \frac{\E\xi}{1-\E\xi} \E_\infty \left( A_1(0) \right) + \E_\infty\left[ A_1(0)^{1/2} \right]. \]
Thus upon a concentration-like result of the kind
\[ \E_\infty[A_1(0)^{1/2}] \approx [\E_\infty(A_1(0))]^{1/2} \]
it is reasonable to expect
\[ \left( 1 - \frac{\E\xi}{1-\E\xi} \right) \E_\infty \left( A_1(0) \right) \approx \left[ \E_\infty(A_1(0)) \right]^{1/2}. \]
Since $1-\E(\xi) / (1-\E\xi) \approx 1-\rho$ this shows that $\E_\infty(A_1(0))$, and hence $\E_\infty(\Lone{\Q(0)})$, should grow as $1/(1-\rho)^2$.
While admittedly crude, the above heuristic arguments provide the
correct estimates, and serve as a useful guide for a rigorous proof
in Section~\ref{infinitebeta}.

As reflected in the above computations, the square factor really stems
from the relation $T^* \approx \tau_{(1)} + \Linfty{\A(0)}^{1/2}$, i.e.,
$T^*$ occurs somehow long after $\tau_{(1)}$, the time at which it
would be optimal to switch in order to avoid inefficient use of the resource.
But it is difficult to make the system switch exactly at $\tau_{(1)}$
in a distributed fashion, and here the penalty incurred is a square root.
Interestingly, the penalty may seem negligible, since it is only
a square root, but this small inefficiency has a significant impact
in heavy traffic.

We explain in the next subsection how we formalize the analysis
of this square root effect.

\subsection{Main results}

In this subsection we present our main results. They are discussed
in the following sections based on a combination of rigorous proofs, heuristic arguments and simulation results, and we explain our contributions in more detail at the end of this subsection. Our first main result introduces the notion of \emph{scaling exponent}, which essentially determines the polynomial growth rate of the mean number of packets in stationarity.

\begin{mr} \label{mr:1}
	$\Q$ is positive recurrent if $\rho < 1$ and transient if $\rho > 1$. Moreover, the limit
	\begin{equation} \label{eq:def-alpha}
		\lim_{\rho \uparrow 1} \left( \frac{\log\E_\infty(\Lone{\Q(0)})}{\log(1/(1-\rho))} \right) \stackrel{\text{\normalfont{def}.}}{=} \alpha
	\end{equation}
	exists and is called scaling exponent.
\end{mr}
\begin{remark}
The transience of $\Q$ when $\rho > 1$ is easy to see. Indeed, if $L(k)$ is the number of packets at the beginning of the $k$th time slot, $L$ is lower bounded by a random walk with drift $R(\rho-1)$. Thus when $\rho > 1$, we have $L(k) \to +\infty$ as $k \to +\infty$ and since $(\Lone{\Q(k)}, k \geq 0)$ is a subsequence of $(L(k), k \geq 0)$ this proves the transience of $\Q$.

Stability for $\rho < 1$ is very intuitive but more challenging to establish. If the active queues are in state $\a = (a_r)$ with $a_r > 0$ for every $r$, the variation of the mean number of packets in the system over the next time slot is equal to
\[ -R\left(1-\rho - z \sum_{r=1}^R \psi(a_r) \right). \]
Since $\psi(a) \to 0$ as $a \to +\infty$, the drift is negative,
close to $-R(1-\rho)$, when each $a_r$ is large enough. The problem
in formalizing this argument is twofold: first, in order to prove stability, one must be able to control every possible initial configuration, not only those where every $a_r$ is large;
second, this argument considers the system on the normal time scale, whereas we are interested in the system embedded at switching times.
\end{remark}
\\

Our interest in the scaling exponent comes from an expected polynomial growth of $\E_\infty(\Lone{\Q(0)})$ in heavy traffic. In general, we expect as $\rho \uparrow 1$ a behavior of the kind
\begin{equation} \label{eq:approx}
	\E_\infty \left( \Lone{\Q(0)} \right) \approx \frac{C}{(1-\rho)^\alpha}
\end{equation}
for some finite constant $C > 0$ (note that this would be stronger than~\eqref{eq:def-alpha}).
The scaling exponent depends on the four model parameters~$R$, $\xi$,
$\zeta$ and~$\beta$.
However, we will mostly be interested in the dependence of~$\alpha$
on~$\beta$ and thus write $\alpha(\beta)$ when the other three
parameters are kept fixed.
In fact, our results suggest that $\alpha$ does not depend on these other
parameters, at least in the two extreme regimes we studied in detail. As explained in Section~\ref{sub:lingering-effect}, the lingering makes $\E_\infty(\Lone{\Q(0)})$ grow as $1/(1-\rho)^2$ when $\beta > 1$, which can be formalized as follows.
\begin{mr} \label{mr:large-beta}
	If $\beta > 1$, then $\alpha(\beta) = 2$.
\end{mr}

It seems quite challenging to prove the existence of, and find a closed-form expression for, $\alpha(\beta)$ when $\beta \leq 1$. Nonetheless, the heuristic arguments of Section~\ref{sub:small-beta} suggest that the
lower bound $\alpha(\beta) \geq 1/\beta$ of~\cite{Bouman11:1} is sharp as $\beta \downarrow 0$, which leads to the following result.
\begin{mr} \label{mr:small-beta}
	$\beta \alpha(\beta) \to 1$ as $\beta \to 0$.
\end{mr}


We now explain in more detail how these three results are established
in the rest of the paper.
The case $\beta = +\infty$ is treated rigorously: positive recurrence
of $\Q$ when $\rho < 1$ and the result that $\alpha(\infty) = 2$ are
proved in Section~\ref{infinitebeta}. For $\beta > 1$, we explain in Section~\ref{sub:perturbation} why it can be seen as a perturbation of the case $\beta = +\infty$: we give some partial results toward a full proof in Sections~\ref{subsub:perturbation} to~\ref{subsub:second-step}, and heuristic arguments in Section~\ref{subsub:intermediate-beta} explaining the technical steps missing for a complete proof. These results are backed up by simulation results in Section~\ref{ec:simula}. Finally, the main result~\ref{mr:small-beta} is discussed based on the simulation results of Section~\ref{sub:simu-beta}, which back up the heuristic arguments of Section~\ref{sub:small-beta}.

\section{Simulation experiments}
\label{ec:simula}

\medskip

In the previous section we provided an overview of the main results
characterizing the heavy traffic behavior of the expected queue lengths.
Before presenting detailed proof arguments in the case $\beta > 1$
in the next section, we first discuss comprehensive simulation
experiments that we conducted to illustrate the stated growth behavior.

The detailed asymptotic analysis and proofs in the next section will
reveal that the value of~$R$ and the precise distributions of~$\xi$
and $\zeta$ do not affect the stability of the system or the value
of the scaling exponent.
Throughout this section we therefore focus on the case where $\xi$ is
geometrically distributed with parameter $2/(2+\rho)$ and $\P(\zeta=1)=1$.
We ran simulations using different distributions for $\xi$
and $\zeta$ as well, including extreme cases such as distributions
with infinite third moment (even infinite second moment for~$\zeta$).
Because of page limitations we do not include these cases,
but they yielded very similar results.

\subsection{Simulations for the main result~\ref{mr:1}} \label{sub:simu-mr1}

\begin{figure}[htbp]
	\vspace{2mm}
	\input{fig/stable}
	\vspace{-7mm}
	\caption{$\Lone{\Q(k)}$ vs.\ $k$ for $R = \beta = 2$ and $\rho=0.99$.}
	\label{fig:stable}
\end{figure}

We now examine the case $R=\beta=2$ in detail. Figures~\ref{fig:stable} and~\ref{fig:instable} show the evolution of $\Lone{\Q(k)}$, starting at $\Q(0)={\bf{0}}$, for $\rho=0.99$ and $\rho=1.01$ respectively.
When $\rho=0.99$, Figure~\ref{fig:stable} shows that $\Lone{\Q(k)}$ fluctuates between $2000$ and $8000$ for $k$ large enough, which strongly suggests that $\Q$ is positive recurrent.
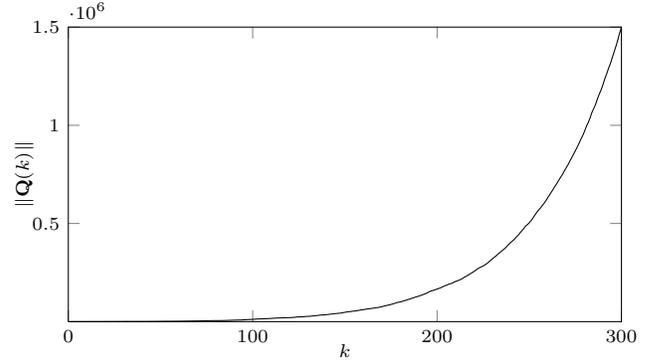
\begin{figure}[htbp]
	\vspace{2mm}
    \begin{tikzpicture}%
\begin{axis}[
font=\scriptsize,
width=\columnwidth+0.5cm, height=5.5cm,
    xmin=0, xmax=300,
    ymin=0, ymax=1500000,
    xtick={0,100,200,300},
    every axis x label/.style=
        {at={(ticklabel cs:0.5)},anchor=center},
    ytick={500000,1000000,1500000},
    every axis y label/.style=
        {at={(ticklabel cs:0.5)},rotate=90,anchor=center},
    scaled ticks=true,
	xlabel={$k$},
	ylabel={$\Lone{\Q(k)}$}]

\addplot[color=black] plot coordinates {
	(1,	3)
	(2,	4)
	(3,	7)
	(4,	8)
	(5,	8)
	(6,	22)
	(7,	22)
	(8,	29)
	(9,	26)
	(10,	33)
	(11,	36)
	(12,	42)
	(13,	78)
	(14,	71)
	(15,	68)
	(16,	71)
	(17,	134)
	(18,	142)
	(19,	127)
	(20,	156)
	(21,	178)
	(22,	216)
	(23,	288)
	(24,	301)
	(25,	312)
	(26,	346)
	(27,	413)
	(28,	439)
	(29,	448)
	(30,	547)
	(31,	660)
	(32,	670)
	(33,	634)
	(34,	584)
	(35,	571)
	(36,	651)
	(37,	560)
	(38,	651)
	(39,	646)
	(40,	636)
	(41,	657)
	(42,	633)
	(43,	759)
	(44,	773)
	(45,	868)
	(46,	924)
	(47,	1017)
	(48,	1111)
	(49,	1206)
	(50,	1289)
	(51,	1372)
	(52,	1481)
	(53,	1555)
	(54,	1666)
	(55,	1656)
	(56,	1729)
	(57,	1860)
	(58,	2044)
	(59,	1921)
	(60,	2046)
	(61,	2184)
	(62,	2337)
	(63,	2467)
	(64,	2726)
	(65,	2801)
	(66,	2901)
	(67,	3031)
	(68,	3398)
	(69,	3548)
	(70,	3596)
	(71,	3774)
	(72,	3896)
	(73,	4239)
	(74,	4359)
	(75,	4631)
	(76,	4818)
	(77,	4958)
	(78,	5136)
	(79,	5293)
	(80,	5242)
	(81,	5170)
	(82,	5354)
	(83,	5397)
	(84,	5579)
	(85,	5890)
	(86,	6150)
	(87,	6547)
	(88,	6774)
	(89,	7268)
	(90,	7416)
	(91,	7593)
	(92,	8133)
	(93,	8152)
	(94,	8569)
	(95,	9103)
	(96,	9644)
	(97,	10388)
	(98,	10946)
	(99,	11339)
	(100,	11638)
	(101,	12319)
	(102,	13221)
	(103,	13605)
	(104,	13865)
	(105,	14011)
	(106,	14668)
	(107,	14912)
	(108,	15336)
	(109,	15813)
	(110,	16563)
	(111,	17236)
	(112,	17712)
	(113,	18320)
	(114,	18595)
	(115,	18902)
	(116,	19218)
	(117,	19668)
	(118,	19769)
	(119,	20250)
	(120,	20397)
	(121,	20860)
	(122,	21583)
	(123,	22213)
	(124,	23674)
	(125,	24053)
	(126,	24483)
	(127,	24412)
	(128,	25356)
	(129,	26406)
	(130,	26769)
	(131,	27380)
	(132,	28322)
	(133,	29261)
	(134,	30350)
	(135,	31067)
	(136,	32241)
	(137,	33685)
	(138,	34342)
	(139,	34689)
	(140,	35605)
	(141,	36498)
	(142,	37997)
	(143,	38741)
	(144,	40088)
	(145,	40966)
	(146,	41913)
	(147,	42786)
	(148,	43730)
	(149,	45888)
	(150,	47529)
	(151,	49102)
	(152,	50427)
	(153,	51239)
	(154,	53031)
	(155,	53840)
	(156,	54896)
	(157,	56609)
	(158,	58248)
	(159,	59968)
	(160,	60861)
	(161,	62770)
	(162,	64263)
	(163,	65437)
	(164,	66408)
	(165,	68090)
	(166,	69326)
	(167,	71014)
	(168,	71976)
	(169,	73988)
	(170,	75768)
	(171,	78521)
	(172,	80788)
	(173,	83174)
	(174,	85041)
	(175,	87504)
	(176,	90052)
	(177,	92788)
	(178,	96663)
	(179,	98488)
	(180,	100382)
	(181,	102839)
	(182,	105847)
	(183,	108469)
	(184,	111657)
	(185,	114855)
	(186,	117736)
	(187,	120738)
	(188,	124201)
	(189,	126792)
	(190,	129539)
	(191,	133020)
	(192,	136193)
	(193,	140412)
	(194,	143250)
	(195,	148070)
	(196,	153818)
	(197,	156656)
	(198,	159577)
	(199,	162521)
	(200,	166320)
	(201,	169054)
	(202,	172797)
	(203,	177497)
	(204,	181783)
	(205,	184934)
	(206,	189504)
	(207,	192655)
	(208,	194800)
	(209,	200469)
	(210,	204820)
	(211,	207619)
	(212,	211771)
	(213,	215436)
	(214,	221292)
	(215,	227574)
	(216,	231854)
	(217,	238061)
	(218,	243814)
	(219,	249627)
	(220,	255146)
	(221,	262218)
	(222,	269225)
	(223,	274635)
	(224,	278826)
	(225,	283391)
	(226,	289353)
	(227,	295643)
	(228,	303093)
	(229,	312041)
	(230,	319593)
	(231,	327680)
	(232,	334730)
	(233,	344209)
	(234,	351229)
	(235,	359008)
	(236,	367073)
	(237,	374835)
	(238,	384860)
	(239,	394862)
	(240,	403687)
	(241,	411425)
	(242,	420217)
	(243,	431386)
	(244,	439844)
	(245,	452074)
	(246,	463150)
	(247,	473745)
	(248,	486347)
	(249,	494157)
	(250,	503491)
	(251,	514722)
	(252,	527623)
	(253,	542531)
	(254,	555614)
	(255,	567728)
	(256,	577943)
	(257,	590557)
	(258,	602524)
	(259,	615807)
	(260,	630650)
	(261,	645689)
	(262,	660276)
	(263,	674645)
	(264,	690129)
	(265,	703947)
	(266,	720377)
	(267,	735135)
	(268,	749188)
	(269,	767612)
	(270,	783367)
	(271,	798679)
	(272,	817351)
	(273,	834454)
	(274,	852642)
	(275,	870408)
	(276,	888906)
	(277,	908203)
	(278,	928879)
	(279,	947594)
	(280,	970202)
	(281,	994840)
	(282,	1013908)
	(283,	1038332)
	(284,	1065879)
	(285,	1086260)
	(286,	1107550)
	(287,	1135720)
	(288,	1159838)
	(289,	1181272)
	(290,	1208355)
	(291,	1237700)
	(292,	1262165)
	(293,	1290590)
	(294,	1314187)
	(295,	1344119)
	(296,	1373582)
	(297,	1402433)
	(298,	1434470)
	(299,	1469659)
	(300,	1499164)
	(301,	1532071)
	(302,	1562530)
	(303,	1596486)
	(304,	1630586)
	(305,	1663557)
	(306,	1701019)
	(307,	1742116)
	(308,	1782502)
	(309,	1820449)
	(310,	1859605)
	(311,	1901035)
	(312,	1937809)
	(313,	1985254)
	(314,	2027700)
	(315,	2066534)
	(316,	2108573)
	(317,	2153965)
};
\end{axis}
\end{tikzpicture}%
	\vspace{-7mm}
	\caption{$\Lone{\Q(k)}$ vs.\ $k$ for $R = \beta = 2$ and $\rho=1.01$.}
	\label{fig:instable}
\end{figure}

When $\rho=1.01$, Figure~\ref{fig:instable} shows that $\Lone{\Q(k)}$ increases until we stop the simulation when $1.5\cdot 10^6$ packets are present in the system. Note that for a transient system we would expect that, when the queues are large, the total queue size grows by a constant amount on average in every time slot and that the time between two consecutive switching times increases. This explains the super-linear growth of $\Lone{\Q(k)}$ as we consider the system at switching times. In fact, the reasoning in Section~\ref{sub:lingering-effect} suggests that when $\Lone{\Q(k)}$ is large we have $\Lone{\Q(k+1)} / \Lone{\Q(k)} \approx \E\xi/(1-\E\xi)$. This gives a heuristic explanation for the exponential growth (at rate $\rho / (2-\rho) > 1$) observed in Figure~\ref{fig:instable}.  All in all, this suggests that $\Q$ is transient for $\rho=1.01$.

Observe moreover that it may in general be difficult to distinguish between
positive recurrent and transient systems based on simulation results.
Here however, $\Q$ obeys two clearly distinguishable types of behavior:
stochastic fluctuations when $\rho < 1$,
and almost deterministic exponential growth when $\rho > 1$.

\subsection{Simulations for the main results~\ref{mr:large-beta} and~\ref{mr:small-beta}} \label{sub:simu-beta}

\begin{figure}[htbp]
	\vspace{2mm}
	\begin{tikzpicture}%
\begin{axis}[
font=\scriptsize,
width=\columnwidth+1cm, height=5.5cm,
    xmin=2, xmax=8,
    ymin=1.5, ymax=2,
    xtick={2,4,6,8},
    every axis x label/.style=
        {at={(ticklabel cs:0.5)},anchor=center},
    scaled ticks=false,
	xlabel={$\log(1/(1-\rho))$},
    legend cell align=left,
    legend style={at={(1,1)}, anchor=north east, inner xsep=1pt, inner ysep=1pt}]

\legend{
	$F(\rho,2)$\\
	Regression result\\%
}
\addplot[color=black] plot coordinates {
	(2.014903,	1.579890)
	(2.420368,	1.643770)
	(2.825833,	1.691541)
	(3.231298,	1.733996)
	(3.636763,	1.768534)
	(4.042229,	1.786989)
	(4.447694,	1.817948)
	(4.853159,	1.831011)
	(5.258624,	1.855233)
	(5.664089,	1.867044)
	(6.069554,	1.867002)
	(6.475019,	1.874079)
};
\addplot[color=gray,densely dashed] plot coordinates {
	(2.014903,	1.610193)
	(2.420368,	1.675226)
	(2.825833,	1.721597)
	(3.231298,	1.756330)
	(3.636763,	1.783319)
	(4.042229,	1.804893)
	(4.447694,	1.822534)
	(4.853159,	1.837227)
	(5.258624,	1.849654)
	(5.664089,	1.860302)
	(6.069554,	1.869527)
	(6.475019,	1.877597)
	(6.880484,	1.884716)
	(7.285949,	1.891043)
	(7.691415,	1.896702)
	(8.096880,	1.901795)
	(8.502345,	1.906402)
	(8.907810,	1.910589)
	(9.313275,	1.914412)
	(9.718740,	1.917916)
};
\end{axis}
\end{tikzpicture}%
	\vspace{-7mm}
	\caption{Approximating $\alpha(2)$ for $R = 2$.}
	\label{fig:beta2toalpha}
\end{figure}
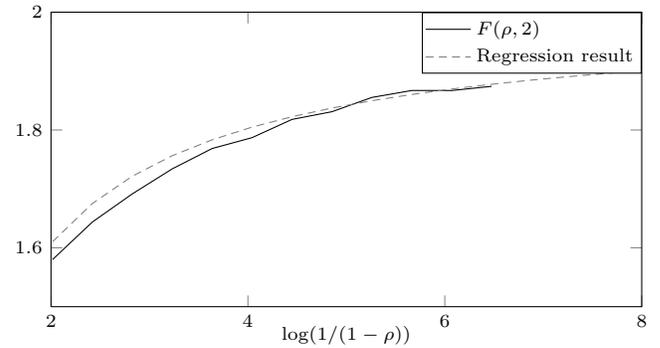

From the results in Figure~\ref{fig:stable} we find by averaging over time $\E_\infty(\Lone{\Q(0)})\approx 4700$ for $\rho = 0.99$, corresponding in view of the definition~\eqref{eq:def-alpha} of the scaling exponent to the estimate $\alpha(2) \approx \log 4700/\log 100 \approx 1.84$. To simplify the discussion define in the sequel
\begin{equation} \label{eq:F}
	F(\rho, \beta) = \frac{\log \E_\infty(\Lone{\Q(0)})}{\log(1/(1-\rho))}, \quad \beta \in (0,\infty], \rho < 1.
\end{equation}
Then, the scaling exponent $\alpha(\beta)$ is defined
in~\eqref{eq:def-alpha} via a limiting procedure, namely
$\alpha(\beta) = F(1-,\beta) = \lim_{\rho \uparrow 1} F(\rho, \beta)$.
Using the results of Figure~\ref{fig:stable} to estimate $\alpha(2)$
amounts to using the approximation $F(1-,2) \approx F(0.99,2)$.
In order to check whether this approximation is valid,
we performed the same simulation for different values of~$\rho$.
In Figure~\ref{fig:beta2toalpha} the value of $F(\rho, 2)$ is plotted
for different values of $\rho \in (0.87,0.999)$. $F(\rho, 2)$ is
plotted versus $\log(1/(1-\rho))$ in order to ``dilate'' time
around the value $\rho = 1$ that we are interested in
(the quantity $\beta \alpha(\beta)$ is plotted versus~$1/\beta$
in Figure~\ref{fig:betaalpha} for the same reason),
and also because it is natural to regress $F(\rho, 2)$, as function
of~$\rho$, against $\log(1/(1-\rho))$ (see forthcoming discussion).

Figure~\ref{fig:beta2toalpha} shows that the limit $F(1-,2)$ seems to exist, but $F(\rho,2)$ is still significantly increasing for $\rho = 0.999$. Thus $F(0.999, 2)$, and in particular $F(0.99, 2)$, cannot be used as an estimate of $\alpha(2)$.
It is numerically difficult to run a simulation for even higher values of~$\rho$, and so to circumvent this problem we use the simulation results displayed in Figure~\ref{fig:beta2toalpha} to find the asymptotic value of $F(\rho, 2)$. To do so, we use the approximation~\eqref{eq:approx} to infer the form of $F(\rho, \beta)$, namely
\begin{equation} \label{eq:regression}
	F(\rho,\beta) \approx \alpha(\beta)  + \frac{\log C}{\log(1/(1-\rho))},
\end{equation}
which suggests, as mentioned above, to regress $F(\rho, 2)$ against $a + b/x$ in the scale $x = \log(1/(1-\rho))$. We performed this regression for the curve displayed in Figure~\ref{fig:beta2toalpha} and found an optimal value $a = 1.9984 \approx \alpha(2)$.

\begin{figure}[htbp]
	\vspace{2mm}
    \begin{tikzpicture}%
\begin{axis}[
font=\scriptsize,
width=\columnwidth+1cm, height=6.5cm,
    xmin=0, xmax=5,
    ymin=0, ymax=10,
    every axis x label/.style=
        {at={(ticklabel cs:0.5)},anchor=center},
    scaled ticks=false,
	xlabel={$\beta$},
    legend cell align=left,
    legend style={at={(1,1)}, anchor=north east, inner xsep=1pt, inner ysep=1pt}]]

\legend{%
	Estimated $\alpha(\beta)$ for $R=2$\\
	Estimated $\alpha(\beta)$ for $R=3$\\
	Estimated $\alpha(\beta)$ for $R=5$\\
	$1/\beta$\\%
}%
\addplot[color=black] plot coordinates {
	(0.100000,	9.996153)
	(0.200000,	4.961067)
	(0.250000,	3.997401)
	(0.300000,	3.326422)
	(0.350000,	2.830031)
	(0.400000,	2.464113)
	(0.450000,	2.188873)
	(0.500000,	1.980168)
	(0.550000,	1.825812)
	(0.600000,	1.721273)
	(0.800000,	1.505995)
	(0.900000,	1.595980)
	(0.950000,	1.705807)
	(0.980000,	1.749895)
	(1.000000,	1.730729)
	(1.000000,	1.748941)
	(1.020000,	1.831359)
	(1.050000,	1.883366)
	(1.100000,	1.860946)
	(1.200000,	1.936386)
	(1.600000,	1.985616)
	(2.000000,	1.998408)
    (3.000000,  2.004374)
	(5.000000,	1.995939)
};
\addplot[color=red] plot coordinates {
	(0.100000,	9.991938)
	(0.200000,	4.968370)
	(0.250000,	3.984266)
	(0.300000,	3.313458)
	(0.350000,	2.872340)
	(0.400000,	2.548786)
	(0.450000,	2.307618)
	(0.500000,	2.127409)
	(0.550000,	1.982591)
	(0.600000,	1.787239)
	(0.600000,	1.765567)
	(0.700000,	1.662231)
	(0.800000,	1.699948)
	(0.800000,	1.706641)
	(0.900000,	1.822516)
	(0.950000,	1.816260)
	(1.000000,	1.813378)
	(1.000000,	1.870033)
	(1.050000,	1.926743)
	(1.100000,	1.969639)
	(1.200000,	1.948308)
	(1.600000,	2.004670)
	(2.000000,	1.998570)
    (3.000000,  1.993457)
	(5.000000,	1.991249)
};
\addplot[color=green] plot coordinates {
	(0.100000,	9.996092)
	(0.200000,	4.993194)
	(0.250000,	4.077599)
	(0.300000,	3.360806)
	(0.350000,	2.852581)
	(0.400000,	2.514373)
	(0.450000,	2.275854)
	(0.500000,	2.112647)
	(0.550000,	1.994263)
	(0.600000,	1.819149)
	(0.700000,	1.656443)
	(0.800000,	1.808673)
	(0.900000,	1.817516)
	(0.950000,	1.762624)
	(1.000000,	1.940740)
	(1.050000,	1.890330)
	(1.100000,	1.934851)
	(1.200000,	1.997691)
	(1.600000,	1.981375)
	(2.000000,	2.004651)
    (3.000000,  1.995208)
	(5.000000,	1.967036)
};
\addplot[color=gray,densely dashed] plot coordinates {
	(0.100000,	10.000000)
	(0.200000,	5.000000)
	(0.250000,	4.000000)
	(0.300000,	3.333333)
	(0.350000,	2.857143)
	(0.400000,	2.500000)
	(0.450000,	2.222222)
	(0.500000,	2.000000)
	(0.550000,	1.818182)
	(0.600000,	1.666667)
	(0.700000,	1.428571)
	(0.800000,	1.250000)
	(0.900000,	1.111111)
	(0.950000,	1.052632)
	(1.000000,	1.000000)
	(1.050000,	0.952381)
	(1.100000,	0.909091)
	(1.200000,	0.833333)
	(1.600000,	0.625000)
	(2.000000,	0.500000)
    (3.000000,  0.333333)
	(5.000000,	0.200000)
};
\end{axis}
\end{tikzpicture}%
	\vspace{-7mm}
	\caption{$\alpha(\beta)$ vs.\ $\beta$ for $R \in \{2,3,5\}$.}
	\label{fig:betatoalpha}
\end{figure}
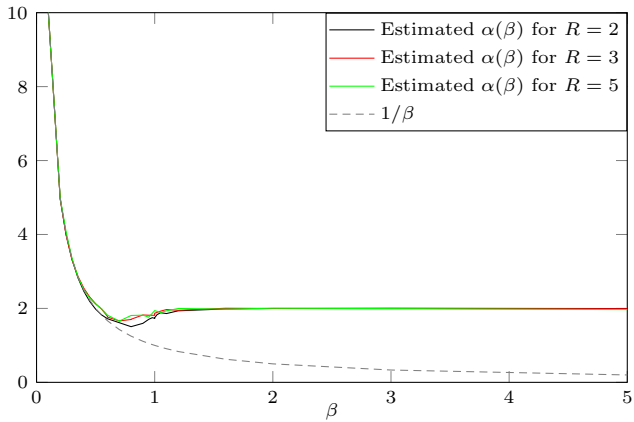

Applying the same approach, we can find an estimate for $\alpha(\beta)$
for any value of~$\beta$.
The results for $R=2$, 3 and~5 are given in Figure~\ref{fig:betatoalpha},
and confirm that $R$ does not seem to influence $\alpha(\beta)$.
Further, the approximation $\alpha(\beta) \approx 2$ appears to be very
good for any $\beta > 1.2$, namely, the estimated value of $\alpha(\beta)$
is at most $3\%$ away from~2 for any $\beta > 1.2$ and any $R = 2,3,5$,
in very good agreement with the main result~\ref{mr:large-beta}.
In Section~\ref{subsub:intermediate-beta} we will discuss in greater
detail the case where $\beta > 1$ is close to~1.

\begin{figure}[htbp]
	\vspace{2mm}
    \begin{tikzpicture}%
\begin{axis}[
font=\scriptsize,
width=\columnwidth+0.9cm, height=5.5cm,
    xmin=1, xmax=10,
    ymin=0, ymax=2,
    every axis x label/.style=
        {at={(ticklabel cs:0.5)},anchor=center},
    scaled ticks=false,
	xlabel={$1/\beta$},
    legend cell align=left,
    legend style={at={(1,1)}, anchor=north east, inner xsep=1pt, inner ysep=1pt}]]

\legend{%
	Estimated $\beta \alpha(\beta)$ for $R=2$\\
	Estimated $\beta \alpha(\beta)$ for $R=3$\\
	Estimated $\beta \alpha(\beta)$ for $R=5$\\%
}%
\addplot[color=black] plot coordinates {
	(10.000000,	0.999615)
	(5.000000,	0.992213)
	(4.000000,	0.999350)
	(3.333333,	0.997926)
	(2.857143,	0.990511)
	(2.500000,	0.985645)
	(2.222222,	0.984993)
	(2.000000,	0.990084)
	(1.818182,	1.004197)
	(1.666667,	1.032764)
	(1.250000,	1.204796)
	(1.111111,	1.436382)
	(1.052632,	1.620517)
	(1.020408,	1.714897)
	(1.000000,	1.730729)
	(1.000000,	1.748941)
	(0.980392,	1.867986)
	(0.952381,	1.977535)
	(0.909091,	2.047041)
	(0.833333,	2.323663)
	(0.625000,	3.176985)
	(0.500000,	3.996816)
    (0.333333,  6.013122)
	(0.200000,	9.979697)
};
\addplot[color=red] plot coordinates {
	(10.000000,	0.999194)
	(5.000000,	0.993674)
	(4.000000,	0.996066)
	(3.333333,	0.994038)
	(2.857143,	1.005319)
	(2.500000,	1.019514)
	(2.222222,	1.038428)
	(2.000000,	1.063704)
	(1.818182,	1.090425)
	(1.666667,	1.072344)
	(1.666667,	1.059340)
	(1.428571,	1.163562)
	(1.250000,	1.359958)
	(1.250000,	1.365313)
	(1.111111,	1.640265)
	(1.052632,	1.725447)
	(1.000000,	1.813378)
	(1.000000,	1.870033)
	(0.952381,	2.023080)
	(0.909091,	2.166603)
	(0.833333,	2.337969)
	(0.625000,	3.207472)
	(0.500000,	3.997140)
    (0.333333,  5.980371)
	(0.200000,	9.956246)
};
\addplot[color=green] plot coordinates {
	(10.000000,	0.999609)
	(5.000000,	0.998639)
	(4.000000,	1.019400)
	(3.333333,	1.008242)
	(2.857143,	0.998403)
	(2.500000,	1.005749)
	(2.222222,	1.024134)
	(2.000000,	1.056323)
	(1.818182,	1.096845)
	(1.666667,	1.091490)
	(1.428571,	1.159510)
	(1.250000,	1.446938)
	(1.111111,	1.635764)
	(1.052632,	1.674492)
	(1.000000,	1.940740)
	(0.952381,	1.984847)
	(0.909091,	2.128336)
	(0.833333,	2.397229)
	(0.625000,	3.170200)
	(0.500000,	4.009301)
    (0.333333,  5.985624)
	(0.200000,	9.835181)
};
\end{axis}
\end{tikzpicture}%
	\vspace{-7mm}
	\caption{$\beta\alpha(\beta)$ vs.\ $1/\beta$ for $R \in \{2,3,5\}$.}
	\label{fig:betaalpha}
\end{figure}
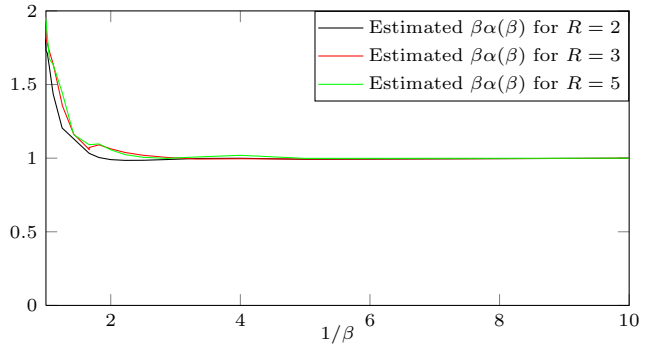

Finally, for small values of~$\beta$, we observe that $\alpha(\beta)$
is close to the lower bound~$1/\beta$.
To see how close $\alpha(\beta)$ is to the lower bound, we made a plot
for $\beta\alpha(\beta)$ in Figure~\ref{fig:betaalpha}.
We observe that $\beta\alpha(\beta) \in (0.98,1.02)$ for any $\beta<0.45$,
which supports our claim made in the main result~\ref{mr:small-beta} and the heuristic argument given in Section~\ref{sub:small-beta}.

\section{Analysis of the case $\beta \in (1,\infty]$}
\label{asan}

\subsection{The case $\beta = +\infty$}
\label{infinitebeta}

We study here in detail the case $\beta {=} \infty$.
In Section~\ref{sub:perturbation} we explain why the case $\beta > 1$
can be seen as a perturbation of this case.
In the remainder of this subsection we assume that $\beta {=} \infty$, so
that $\A(T^*) = {\bf{0}}$ and $T^* = \inf\{k \geq 0: \A(k) = {\bf{0}}\}$.

\subsubsection{Control of $T^*$}

We first prove that $T^* \approx \tau_{\max}$, i.e., the time at which
the $R$~independent random walks $A_r$ simultaneously hit~0, is close to
the first time at which each process has visited~0 at least once.

\begin{lemma}\label{lemma:T^*}
	$\sup_{\a, \rho} \E_\a\left(T^* - \tau_{\max} \right)$ is finite, where the supremum is taken over $\a \in \N^R$ and $\rho \leq 1$.
\end{lemma}
\begin{proof}
	By monotonicity of~$T^*$ in~$\rho$, it is enough to show
that $\sup_\a \E_\a(T^* - \tau_{\max})$ is finite for $\rho = 1$.
So assume in the remainder of the proof $\rho = 1$: then each active queue is stable (for that we only need $\rho < 2$) and in particular, $T^*$ is almost surely finite. Thus the strong Markov property at time $\tau_{\max} \leq T^*$ gives
	\[ \E_\a\left(T^* - \tau_{\max} \right) = \E_\a \left[ \E_{\A(\tau_{\max})}(T^*) \right]. \]
	When $A_r$ hits $0$ for the first time at time $\tau_r$, we can couple it with a stationary version $\tilde A_r$ of $A_1$ in such a way that $\tilde A_r(k) \geq A_r(k)$ for every $k \geq \tau_r$ (this is the usual coupling between two processes starting from different initial conditions and sharing the same stochastic primitives). By monotonicity and since queues are independent, we obtain
	\[ \E_\a\left(T^* - \tau_{\max} \right) \leq \E_\varrho(T^*), \]
	where for any $r \geq 2$, we define $\varrho = (X_1, \ldots, X_R) \in \N^R$ with $(X_i, i \geq 1)$ i.i.d.\ with common distribution the stationary distribution $X$ of $A_r$.

	Define $\sigma(0) = 0$ and for $k \geq 1$, $\tau(k) = \tau_{\max} \circ \theta_{\sigma(k-1)}$ and $\sigma(k) = \sigma(k-1) + \tau(k)$, where $\theta$ is the shift operator. In words, $\sigma(k+1)$ is the smallest time after $\sigma(k)$ such that every queue will have visited~0 at least once after (and including) time $\sigma(k)$. Note that if $\A(\sigma(k)) = {\bf{0}}$ for some $k$, then $\sigma(k') = \sigma(k)$ for $k' \geq k$. In particular, the limit $\sigma(\infty)$ as $k \to +\infty$ exists, and is by construction equal to~$T^*$. Since $\tau(k) = \sigma(k) - \sigma(k-1)$ this can be rewritten as $T^* = \sum_{k \geq 1} \tau(k)$.
	
	By construction, $A_r$ takes the value~0 at least once between $\sigma(k)$ and $\sigma(k+1)$: at this time we can couple it with a stationary version of itself that stays above it. Following this stationary process from this time until the time $\sigma(k+1)$, this gives the existence of i.i.d.\ random variables $(X_r(k), r \in \Rcal, k \geq 0)$ with common distribution $X$ such that $A_r(\sigma(k)) \leq X_r(k)$ for every $r$ and $k$.

	Considering the time needed for every queue to visit~0 at least once starting from $(X_r(k), r \in \Rcal)$ gives an upper bound on $\tau(k)$. Each time we do so, there is a probability $p = \P_\varrho(\A(\tau_{\max}) = {\bf{0}}) \in (0,1)$ that we reach ${\bf{0}}$ at time $\tau_{\max}$. Thus, $T^*$ is (stochastically) upper bounded by $\tau'_1 + \cdots + \tau'_G$ with $G$ a geometric random variable with parameter $1-p$ and $(\tau'_i)$ i.i.d., independent from $G$ and with common distribution $\tau_{\max}$ under $\P_\varrho$. This implies
	\[ \E_\varrho\left( T^* \right) \leq \E \left( \tau'_1 + \cdots + \tau'_G \right) = \frac{\E_\varrho \left( \tau_{\max} \right)}{\P_\varrho \left( \A(\tau_{\max}) = {\bf{0}} \right)}. \]
	Since $\tau_{\max} \leq \tau_1 + \cdots + \tau_R$, $\E_{a_r}(\tau_r) = a_r/(1-\E\xi)$ and $X$ has finite mean (because $\xi$ is assumed to have finite second moment), $\E_\varrho(\tau_{\max})$ is finite and so the proof is complete.
\end{proof}

The previous lemma justifies the approximation $T^* \approx \tau_{\max}$,
and to further control $\tau_{\max}$, we will use the fact that
$(\tau_r, r \in \Rcal)$ under $\P_\a$ is equal in distribution
to $(V_r(a_r), r \in \Rcal)$, where $(V_r, r \in \Rcal)$ are i.i.d.\
random walks started at~0 and with step distribution~$\delta$,
equal in distribution to $\tau_1$ under $\P_1$ (i.e., $\delta$ is the time
needed for a random walk with step distribution $\xi-1$ to go from~1 to~0).
Then $\delta$ has finite mean $1/(1-\E\xi)$ and also finite variance,
which we denote by~$\nu$.
To control the maximum of random walks, we will use the following result.

\begin{lemma} \label{lemma:bound-max}
	If $(W_r)$ are i.i.d.\ random walks started at~0 with step distribution having mean $m$ and variance $w$, then for any $\x = (x_r) \in \N^R$ it holds that
	\[ \E \left( \max_r W_r(x_r) \right) \leq m \Linfty{\x} + R (w \Linfty{\x})^{1/2}. \]
\end{lemma}

\begin{proof}
	Defining $Y_r = (W_r(x_r) - m x_r) / (w x_r)^{1/2}$, we have $\max_r W_r(x_r) = \max_r (m x_r + (w x_r)^{1/2} Y_r)$ so that
	\[ \max_r W_r(x_r) 
	\leq m \Linfty{\x} + (w \Linfty{\x})^{1/2} \Lone{{\bf{Y}}} \]
	which proves the result since $\E(\Lone{{\bf{Y}}}) = \sum_r \E(\Linfty{Y_r}) \leq R$.
\end{proof}

\subsubsection{Proof of $\alpha(\infty) = 2$} Before proving the main result~\ref{mr:large-beta} in the case $\beta = +\infty$, we first need to prove stability and finiteness of the stationary mean.

\begin{prop} \label{prop:stab}
	Assume that $\beta = +\infty$ and that $\rho < 1$. Then $\Q$ is positive recurrent and $\E_\infty(\Lone{\Q(0)}) < +\infty$.
\end{prop}

\begin{proof}
	Let $\Phi(\q) = |\q^a| + |\q^i|$: to prove Proposition~\ref{prop:stab}, it is enough to prove that
	\begin{equation} \label{eq:geom-lyapunov}
		\lim_{K \to +\infty} \sup_{\q: \Phi(\q) \geq K} \left( \frac{\E_\q[\Phi(\Q(2)) - \Phi(\q)]}{\Phi(\q)} \right) < 0.
	\end{equation}
	Indeed, this shows that $\Phi$ is a Lyapunov function,
which implies positive recurrence of $\Q$ using for instance the
Foster-Lyapunov criterion.
But it shows more than that: in the terminology of~\cite{Gamarnik06:0}
it implies that $\Phi$ is a geometric Lyapunov function,
and Theorem~5 in~\cite{Gamarnik06:0} shows
that~\eqref{eq:geom-lyapunov} implies that $\E_\infty[\Phi(\Q(0))]$,
and in particular $\E_\infty(\Lone{\Q(0)})$, is finite.
Thus we only have to prove~\eqref{eq:geom-lyapunov}.

	Since $\Linfty{\x + {\bf{y}}} \leq \Linfty{\x} + \Linfty{{\bf{y}}}$,~\eqref{eq:Q} implies that
	\[ \E_\q\left[\Phi(\Q(1))\right] \leq \Linfty{\q^i} + \E_{\q^a}\left(\Linfty{\S(T^*)}\right) + \E_{\q^a}\left(\Linfty{\A(T^*)}\right). \]
	Since $\S$ and $T^*$ are independent, Lemma~\ref{lemma:bound-max} gives that $\E_\a(\Linfty{\S(T^*)}) \leq \E(\xi) \E_\a(T^*) + R \E_\a((vT^*)^{1/2})$ for any $\a \in \N^R$ (recall that $v$ is the variance of~$\xi$, assumed to be finite). Thus after rearranging the terms, we end up with the bound
	\begin{equation} \label{eq:bound-drift}
		\E_\q\left[\Phi(\Q(1)) - \Phi(\q)\right] \leq
-(1-\rho) \E_{\q^a}(T^*) + \Psi(\q^a),
	\end{equation}
	where
	\[ \Psi(\a) = \E_{\a}\left( \Linfty{\A(T^*)} + (1-\E\xi) T^*- \Linfty{\A(0)} + R (vT^*)^{1/2} \right). \]
	We now argue that $\Psi(\a) \leq c [\E_{\a}(T^*)]^{1/2}$ for some finite constant $c$ independent of $\a$. Since we are considering the case $\beta = +\infty$, we have $|\A(T^*)| = 0$ and so we only have to show that $(1-\E\xi) \E_{\a}(T^*) - \Linfty{\a} \leq c[\E_{\a}(T^*)]^{1/2}$. By Lemma~\ref{lemma:T^*}, we have $\E_{\a}(T^*) \leq \E_{\a}(\tau_{\max}) + c'$ for some finite constant $c'$ independent of $\a$. Further, since $\tau_{\max}$ is equal in distribution to $\max_r V_r(a_r)$, Lemma~\ref{lemma:bound-max} gives (recall that $\nu$ is the variance of the step distribution of the $V_r$'s)
	\[ \E_\a(T^*) \leq \frac{\Linfty{\a}}{1-\E\xi} + R (\nu \Linfty{\a})^{1/2} + c', \]
	which implies the existence of the desired constant $c$ such that $\Psi(\a) \leq c [\E_{\a}(T^*)]^{1/2}$. Defining
	\begin{equation} \label{eq:def-Phi}
		\Gamma(\a) = (1-\rho) \E_\a\left(T^*\right) - c \left[\E_\a\left(T^*\right)\right]^{1/2},
	\end{equation}
	\eqref{eq:bound-drift} can be rewritten as $\E_\q\left[\Phi(\Q(1)) - \Phi(\q)\right] \leq -\Gamma(\q^a)$. Using the Markov property and~\eqref{eq:Q}, this gives
	\begin{multline} \label{eq:bound-drift-2}
		\E_\q\left[\Phi(\Q(2)) - \Phi(\q)\right]\\
		\leq - \E_{\q^a} \left( \Gamma(\q^a) + \Gamma(\q^i + \S(T^*)) \right).
	\end{multline}
	When $\Phi(\q) = \Linfty{\q^a} + \Linfty{\q^i}$ is large, at least one of the $2R$ coordinates of $\q$ must be large. Since $\E_\a(T^*) \geq a_r/(1-\E\xi)$ (as a consequence of $T^* \geq \tau_r$ and $\E_a(\tau_r) = a/(1-\E\xi)$), it is not hard to show that
	\[ \lim_{K \to +\infty} \inf_{\q: \Phi(\q) \geq K} \frac{\E_{\q^a} \left( \Gamma(\q^a) + \Gamma(\q^i + \S(T^*)) \right)}{\Phi(\q)} > 0, \]
	which completes the proof of the result.
\end{proof}
	
Now that we have stability and finiteness of the stationary mean, we prove that $\alpha(\infty) = 2$. The proof will make use of the following result.

\begin{lemma} \label{lemma:bound-max-2}
	If $(W_j)$ are $J$ i.i.d.\ random walks started at~0 with non-negative step distribution having mean $m$ and variance $w$, then for any $\x = (x_j) \in \N^J$ it holds that,
	\[ \E \left[ \left( \max_{j} W_j(x_j) \right)^{1/2} \right] \geq (m \Linfty{\x})^{1/2} - (w/m^{3/2}) \Linfty{\x}^{-1/2}. \]
\end{lemma}

\begin{proof}
	Since $\E(\max_j W_j(x_j)) \geq \max_j \E(W_j(x_j))$ it is enough to prove the result for $J = 1$. Fix $k \geq 0$ and let $Y = (W_1(k) - mk)/(mk)^{1/2}$ and
	\[ f(y) = \frac{1+y/2 - (1+y)^{1/2}}{y^2}, \ y \geq -1. \]
	Since $\E Y = 0$ and $Y \geq -(km)^{1/2}$, we can rewrite after some algebra
	\[ \E \big(W_1(k)^{1/2}\big) = (km)^{1/2} - (km)^{-1/2} \E\left[Y^2 f\big( (km)^{-1/2} Y \big) \right], \]
	and since $\sup f = 1/2$ and $\E(Y^2) = w/m$ this gives the result.
\end{proof}

\begin{theorem}\label{thm:alpha-beta=infty}
	If $\beta = +\infty$, then
	\begin{multline} \label{eq:alpha(infty)}
		0 < \liminf_{\rho \uparrow 1} \left[ (1-\rho)^2 \E_\infty\left(\Lone{\Q(0)}\right) \right]\\
		\leq \limsup_{\rho \uparrow 1} \left[ (1-\rho)^2 \E_\infty\left(\Lone{\Q(0)}\right) \right] < +\infty.
	\end{multline}
	In particular, $\alpha(\infty) = 2$.
\end{theorem}

\begin{proof}
	Since $\beta = +\infty$, we have
	\[ \E_\infty(\Lone{\Q(0)}) = \E_\infty(\Lone{\A(0)}) = R \E_\infty(A_1(0)). \]
	Thus we only need to prove
	\begin{equation} \label{eq:upper-bound}
		\limsup_{\rho \uparrow 1} \left[ (1-\rho)^2 \E_\infty\left(A_1(0)\right) \right] < +\infty,
	\end{equation}
	which implies the lower bound in~\eqref{eq:alpha(infty)}, and
	\begin{equation} \label{eq:lower-bound}
		\liminf_{\rho \uparrow 1} \left\{ (1-\rho) \E_\infty\big[A_1(0)^{1/2}\big] \right\} > 0
	\end{equation}
	which by Jensen's inequality implies the upper bound in~\eqref{eq:alpha(infty)}.
	\\
	
	\noindent \textit{Proof of~\eqref{eq:upper-bound}.} Starting from~\eqref{eq:A_1}, using that $\E_a(\tau_1) = a/(1-\E\xi)$, subtracting on both sides $\E(\xi) \E_\infty(\tau_1)$ (for this precise operation we need the finiteness of the stationary first moment, to avoid doing $\infty-\infty$) and dividing by $\E\xi$, we end up with
	\[ \frac{1}{\E\xi} \left( 1 - \frac{\E\xi}{1-\E\xi} \right) \E_\infty\left(A_1(0)\right) = \E_\infty\left( T^* - \tau_1 \right). \]
	Then, adding and subtracting $\tau_{\max}$ in the right hand side, and using that $\E\xi = \rho/2$, we obtain
	\[ g_\rho (1-\rho) \E_\infty\left(A_1(0)\right) = \E_\infty\left( \tau_{\max} - \tau_1 \right) + \E_\infty\left( T^* - \tau_{\max} \right), \]
	with $g_\rho = 4/(\rho(2-\rho))$. Thus in view of Lemma~\ref{lemma:T^*}, to prove~\eqref{eq:upper-bound} we only have to show that
	\begin{equation} \label{eq:upper-bound-tau}
	\limsup_{\rho \uparrow 1} \left( \frac{\E_\infty\left(\tau_{\max} - \tau_1 \right)}{\left[ \E_\infty(A_1(0)) \right]^{1/2}} \right) < +\infty.
\end{equation}
	Applying Lemma~\ref{lemma:bound-max} to $\tau_{\max}$ under
$\P_\a$ (equal in distribution to $\max_r V_r(a_r)$), we obtain,
denoting temporarily $\mu = 1/(1-\E\xi)$, $\E_\a\left(\tau_{\max} -
\tau_1 \right) \leq \mu (\Linfty{\a} - a_1) + R \nu \Linfty{\a}^{1/2}$.
Integrating over the stationary distribution of $\Q$ and using Jensen's
inequality, we obtain
	\begin{multline*}
		\E_\infty\left(\tau_{\max} - \tau_1 \right) \leq \mu \E_\infty\left(\Linfty{\A(0)} - A_1(0) \right)\\
		+ R \nu \left[ \E_\infty\left(A_1(0) \right) \right]^{1/2}.
	\end{multline*}
	In particular, to prove~\eqref{eq:upper-bound-tau} it is enough to show that
	\begin{equation} \label{eq:upper-bound-diff}
		\limsup_{\rho \uparrow 1} \left( \frac{\E_\infty\left( \Linfty{\A(0)} - A_1(0) \right)}{\left[ \E_\infty(A_1(0)) \right]^{1/2}} \right) < +\infty.
	\end{equation}
	We have already shown in the proof of Proposition~\ref{prop:stab} that
	\[ \E_\infty \left(\Linfty{\A(0)}\right) = \E_\infty \left(\Linfty{\S(T^*)}\right) \leq \E(\xi) \E_\infty(T^*) + R v \left[\E_\infty(T^*)\right]^{1/2}, \]
	and since $\E_\infty(A_1(0)) = \E(\xi) \E_\infty(T^*)$ this gives
\[ \E_\infty \left( \Linfty{\A(0)} - A_1(0) \right) \leq Rv \left[ \E_\infty(A_1(0)) / \E\xi \right]^{1/2}. \]
	This proves~\eqref{eq:upper-bound-diff} which completes the proof of~\eqref{eq:upper-bound}. \\
	
	\noindent\textit{Proof of~\eqref{eq:lower-bound}.} We have
	\[ \E_\infty\big[(Q_1^a(0))^{1/2}\big] = \E_\infty\big[(S_1(T^*))^{1/2}\big] \geq \E_\infty\big[(S_1(\tau_{\max}))^{1/2}\big]. \]
	Applying Lemma~\ref{lemma:bound-max-2} to $S_1(\tau_{\max})$ by using the independence between $S_1$ and $\tau_{\max}$, we obtain
	\begin{multline*}
		\E_\infty\big[(Q_1^a(0))^{1/2}\big] \geq (\E\xi)^{1/2} \E_\infty\big[(\tau_{\max})^{1/2}\big]\\
		- c \E_\infty\left[(\tau_{\max})^{-1/2}\right]
	\end{multline*}
	for some finite constant~$c$ independent of~$\rho$. Applying again Lemma~\ref{lemma:bound-max-2} to $\tau_{\max}$ we obtain
	\begin{multline*}
		\E_\infty\big[(Q_1^a(0))^{1/2}\big] \geq (\mu\E\xi)^{1/2} \E_\infty (\Linfty{\Q^a(0)}^{1/2})\\
		- c \E_\infty (\Linfty{\Q^a(0)}^{-1/2}) - c \E_\infty\left[(\tau_{\max})^{-1/2}\right].
	\end{multline*}
	Subtracting $(\mu\E\xi)^{1/2} \E_\infty (Q_1^a(0)^{1/2})$ on both sides we finally end up with
	\begin{multline*}
		h_\rho (1 - \rho) \E_\infty\big[(Q_1^a(0))^{1/2}\big]\\
		\geq (\mu\E\xi)^{1/2} \E_\infty \left(\Linfty{\Q^a(0)}^{1/2} - (Q_1^a(0))^{1/2}\right)\\
		- c \E_\infty \left(\Linfty{\Q^a(0)}^{-1/2}\right) - c \E_\infty\left[(\tau_{\max})^{-1/2}\right],
	\end{multline*}
	with $h_\rho = (1-(\rho/(2-\rho))^{1/2})/(1-\rho) \to 1$ as $\rho \to 1$. The two last terms of the previous lower bound vanish as $\rho \uparrow 1$. As for the first term, it is not hard based on~\eqref{eq:Q} to show that the $R$-dimensional vector $((Q_1^a(0)+1)^{-1/2} (Q_r^a(0) - Q_1^a(0)), r \in \Rcal)$ under $\P_\infty$ converges weakly as $\rho \uparrow 1$, from which one readily deduces thanks to the continuous mapping theorem that $\Linfty{\Q^a(0)}^{1/2} - (Q_1^a(0))^{1/2}$ also converges weakly as $\rho \uparrow 1$ to a random variable which is not identically zero. Using Fatou's lemma, this implies that
	\[ \liminf_{\rho \uparrow 1} \E_\infty \left(\Linfty{\Q^a(0)}^{1/2} - (Q_1^a(0))^{1/2}\right) > 0, \]
	which concludes the proof.
\end{proof}

\subsection{The case $\beta > 1$} \label{sub:perturbation}

\subsubsection{Main result and a perturbation argument} \label{subsub:perturbation}

In the previous subsection, we analyzed in detail the case
$\beta = \infty$ and proved $\alpha(\infty) = 2$.
On the other hand, the simulation results,
see Figure~\ref{fig:betatoalpha}, strongly suggest that
$\alpha(\beta) = 2$ for any $\beta > 1$,
and the goal of this subsection is to explain this result.
The key to understand the behavior of $\Q$ for $\beta > 1$ is
Proposition~\ref{prop:T_1}.

This proposition~\ref{prop:T_1} is only concerned with the behavior of an active queue, i.e., a queue subject to the dynamic~\eqref{eq:A}. Thus in order to prove Proposition~\ref{prop:T_1}, we only need to assume $\rho < 2$, which is the stability condition for an active queue. It is easy to see that these results actually hold uniformly in $\rho \leq 1$, which is what would be needed in order to go in heavy traffic for the full model.

In the sequel, we let $T_1$ be the time at which $A_1$ advertizes a release for the first time.
\begin{prop} \label{prop:T_1}
	If $\beta > 1$, then $(A_1(T_1), T_1 - \tau_1)$ under $\P_{a_1}$ converges weakly as $a_1 \to +\infty$ to a non-degenerate random variable.
\end{prop}
By non-degenerate, we mean a random variable $(X, Y) \in \N \times \Z$ such that both $X$ and $Y$ are non-deterministic and almost surely finite (an explicit expression for the weak limit of $(A_1(T_1), T_1 - \tau_1)$ is given in Lemma~\ref{lemma:limiting-picture}). Note that from Proposition~\ref{prop:T_1}, which is concerned with the behavior of one active queue, one can easily deduce the behavior of $\A(T^*)$ and~$T^*$ (the proof of the following result is only sketched in order to comply with page limitations).
\begin{corollary} \label{cor:T^*}
	Assume that $\beta > 1$ and consider any sequence of initial states $\a_n = (a_{n,r})$ such that $\min_r a_{n,r} \to +\infty$: then $(\A(T^*), T^* - \tau_{\max})$ under $\P_{\a_n}$ converges weakly as $n \to +\infty$ to a non-degenerate random variable.
\end{corollary}
\begin{proof}[(sketch)]
	At time~$T^*$, each active queue needs to have advertized a release at least once. At this time, the active queue was of order one by Proposition~\ref{prop:T_1} and since it is stable, it remains of order one at time~$T^*$.
\end{proof}
In the extreme case $\beta = +\infty$ we have by definition $\A(T^*) = {\bf{0}}$ and (by Lemma~\ref{lemma:T^*}) $T^* \approx \tau_{\max}$. By showing that both $\A(T^*)$ and $T^* - \tau_{\max}$ are of order one when $\beta > 1$, Corollary~\ref{cor:T^*} therefore justifies seeing the case $\beta > 1$ as a perturbation of the case $\beta = +\infty$. In particular, treating $\E_\infty(\A(T^*))$ and $\E_\infty(T^* - \tau_{\max})$ as constants, the heuristic reasoning outlined in Section~\ref{sub:lingering-effect} goes through and again predicts a scaling exponent $\alpha(\beta) = 2$ for any $\beta > 1$. This is indeed in very good agreement with the simulation results discussed in Section~\ref{sub:simu-beta} and was stated as our main result~\ref{mr:large-beta}.

The remainder of this subsection is devoted to the proof of Proposition~\ref{prop:T_1}. The proof proceeds in two steps: in the first step we show that the proof of Proposition~\ref{prop:T_1} reduces to proving a simpler property of some particular random walk (see~\eqref{eq:goal-N}). In the second step we prove that this property holds when $\beta > 1$.

\subsubsection{First step}

Since before time $T_1$, the active queue $A_1$ does by definition not
advertize any release, it is enough to prove Proposition~\ref{prop:T_1} in
the case $\zeta = 0$, which we assume in the remainder of this subsection.
In particular, $A_1$ is a random walk with step distribution $\xi-1$ reflected at~0. The reduction of the proof of Proposition~\ref{prop:T_1} to proving~\eqref{eq:goal-N} relies on the following coupling of the processes $A_1$ for all possible initial states $a \geq 0$.

Let $V$ and $W^\uparrow$ be two independent processes with the following distribution. Let $V$ be a version of $A_1$ under $\P_0$, i.e., $(V(k), k \geq 0)$ is a random walk started at~0, with step distribution $\xi-1$ and reflected at~0.

Let $W^\uparrow$ be a random walk started at~0, with step distribution $1 - \xi$ and conditioned on never visiting~0 after time~0: since $\E(1-\xi) > 0$ this conditioning is well-defined. Let moreover $\kappa(a) = \max\{ k \geq 0: W^\uparrow(k) = a \}$ be the time of the last visit to $a \in \{0,1,\ldots,\infty\}$ (to be understood as $\kappa(a) = +\infty$ for $a = +\infty$), so that $\kappa(a)$ is almost surely finite if $a$ is finite. Let finally $W^\uparrow_a$ be the process $W^\uparrow$ stopped at $\kappa(a)$, i.e., $W^\uparrow_a(k) = W^\uparrow(k)$ if $k \leq \kappa(a)$ and $W^\uparrow_a(k) = W^\uparrow(\kappa(a)) = a$ if $k \geq \kappa(a)$.

\begin{lemma}
	Extend $A_1$ on $\Z$ by setting $A_1(k) = A_1(0)$ for $k \leq 0$, and let $A^+ = (A_1(\tau_1+k), k \geq 0)$ and $A^- = (A_1(\tau_1-k), k \geq 0)$. Then for any finite $a \geq 0$, $(A^+, A^-)$ under $\P_a$ is equal in distribution to $(V, W^\uparrow_a)$.
	
	In particular, as $a \to +\infty$, $(A^+, A^-)$ under $\P_a$ converges weakly to $(V, W^\uparrow)$.
\end{lemma}

\begin{proof}
	That $A^+$ is equal in distribution to $V$ and is independent from $A^-$ follows from the strong Markov property at time~$\tau_1$. That $A^-$ is equal in distribution to $W^\uparrow_a$ comes from duality. The weak convergence result then follows from the fact that $\kappa(a) \to +\infty$ almost surely as $a \to +\infty$, so that $(V, W^\uparrow_a) \to (V, W^\uparrow)$ almost surely as $a \to +\infty$.
\end{proof}

Essentially, this representation of $A_1$ shifts the origin of time at $\tau_1$: $A^+$ looks at $A_1$ from time $\tau_1$ forward in time, while $A^-$ looks at $A_1$ from $\tau_1$ backward in time. Moreover, this representation couples all the processes $A_1$ with different initial states on the same probability space, which yields a simple representation for the law of $(A_1(T_1), T_1 - \tau_1)$. Let in the sequel $Z = (Z(k), k \in \Z)$ be the following process (indexed by $\Z$): $Z(k) = V(k)$ if $k \geq 0$ and $Z(k) = W^\uparrow(-k)$ if $k \leq 0$. The previous coupling immediately implies the following result.

\begin{lemma} \label{lemma:limiting-picture}
	Let $(U_k, k \in \Z)$ be i.i.d., uniformly distributed in $[0,1]$, independent from $Z$, and for $0 \leq a \leq +\infty$ and $k \in \Z$ let
	\[ D_{a,k} = \begin{cases}
		0 & \text{ if } k < -\kappa(a),\\
		\indicator{U_k < \psi(Z(k))} & \ \text{ else,}
	\end{cases} \]
	and $T_a^Z = \inf\left\{ k \in \Z: D_{a,k} = 1 \right\}$. Then for any finite $a \geq 0$, $(A_1(T_1), T_1 - \tau_1)$ under $\P_a$ is equal in distribution to $(Z(T_a^Z), T_a^Z)$ and in particular, it converges weakly as $a \to +\infty$ to $(Z(T_\infty^Z), T_\infty^Z)$ (with $Z(T_\infty^Z) = +\infty$ if $T_\infty^Z = -\infty$).
\end{lemma}

\begin{proof}
The equality in law between $(A_1(T_1), T_1 - \tau_1)$
and $(Z(T_a^Z), T_a^Z)$ is clear from the construction, and not
difficult (although a bit heavy in notation) to formalize.
Moreover, $T_a^Z$ is by construction decreasing in~$a$ and so its limit
as $a \to +\infty$ exists. It is not hard to show that its limit is
exactly $T_\infty^Z$ and by continuity we deduce that
$Z(T_a^Z) \to Z(T_\infty^Z)$ as $a \to +\infty$, which implies the result.
\end{proof}

Thus to prove Proposition~\ref{prop:T_1}, we only have to establish
that $|T_\infty^Z|$ is (almost surely) finite.
Since $V$ is positive recurrent and starts at~0, it is clear that
$\min\left\{ k \geq 0: D_{\infty,k} = 1 \right\}$ is finite and so to
prove that $|T_\infty^Z|$ is finite, we only have to demonstrate
that $\inf\left\{ k \leq 0: D_{\infty,k} = 1 \right\}$ is finite.
In other words, we have to prove that
$\sup \left\{ k \geq 0: U_{-k} < \psi(W^\uparrow(k)) \right\}$ is finite,
which informally means that $W^\uparrow$ advertizes a release only
finitely many times.

So in the sequel, we consider $(U_k, k \geq 0)$ i.i.d.\ random variables,
uniformly distributed on $[0,1]$ and independent of $W^\uparrow$,
and we define
\[ N = \sum_{k \geq 0} \indicator{U_k < \psi(W^\uparrow(k))} \]
as the number of times $W^\uparrow$ advertizes a release. The proof of Proposition~\ref{prop:T_1} will thus be complete if we can prove that
\begin{equation} \label{eq:goal-N}
	\P(N < +\infty) = 1.
\end{equation}

\subsubsection{Second step} \label{subsub:second-step}

We now assume that $\beta > 1$ and we prove that $\P(N > n) \to 0$ as $n \to +\infty$, which will prove~\eqref{eq:goal-N}. By definition,
\[ \P\left( N = 0 \right) = \P \left( U_k > \psi(W^\uparrow(k)), k \geq 0 \right), \]
and since $W^\uparrow$ and the $U_k$'s are independent this gives
\[ \P\left( N = 0 \right) = \E \left[ \prod_{k \geq 0} \big( 1 - \psi(W^\uparrow(k)) \big) \right]. \]
Let $a \geq 0$: introducing $\varphi(a) = -\log(1-\psi(a))$ and $L^\uparrow(a) = \sum_{k \geq 0} \indicator{W^\uparrow(k) = a}$, the local time at level $a$, we obtain
\begin{equation} \label{eq:N=0}
	\P\left( N = 0 \right) = \E \left[ \exp \left( - \sum_{a \geq 0} \varphi(a) L^\uparrow(a) \right) \right].
\end{equation}

\begin{lemma} \label{lemma:>0}
	The quantity $\sup_{a \geq 0} \E(L^\uparrow(a))$ is finite. In particular, $\P(N = 0) > 0$.
\end{lemma}

\begin{proof}
	Let $W^-$ be a random walk with step distribution $\xi-1$, started at~0 and independent from $W^\uparrow$, and for $k \in \Z$ define $W^*(k) = W^\uparrow(k)$ if $k \geq 0$ and $W^*(k) = W^-(-k)$ if $k \leq 0$. Thus, defining $L^*(a) = \sum_{k \in \Z} \indicator{W^*(k) = a}$ we have the obvious inequality $L^\uparrow(a) \leq L^*(a)$ and so we only have to prove that $\sup_{a \in \Z} \E(L^*(a))$ is finite.
	
	It is clear that $L^*$ stays the same if $W^*$ is shifted in time, and that shifting $L^*$ in time amounts to shifting $W^*$ in space. Moreover, for any $w \in \Z$ the process $(W^*(k) + w, k \in \Z)$ shifted at the time of last visit to~0 is equal in distribution to $W^*$. Combining these facts, we see that $L^*$ is a stationary sequence and in particular, $\sup_{a \in \Z} \E(L^*(a)) = \E(L^*(0))$. But by the strong Markov property, it is clear that $L^*(0)$ is a geometric random variable, in particular it has finite first moment. This proves the finiteness of $\sup_a \E(L^\uparrow(a))$.
	
	As for $\P(N > 0)$, we have
	\[ \E \left( \sum_{a \geq 0} \varphi(a) L^\uparrow(a) \right) \leq \sup_a \E(L^\uparrow(a)) \sum_{a \geq 0} \varphi(a), \]
	and since $\varphi(a) \sim a^{-\beta}$ as $a \to +\infty$
and $\beta > 1$, the sum $\sum_a \varphi(a)$ is finite which implies,
in view of the last display, that the random variable
$\sum_a \varphi(a) L^\uparrow(a)$ is almost surely finite.
This proves $\P(N = 0) > 0$ in view of~\eqref{eq:N=0} and concludes
the proof of the lemma.
\end{proof}

We now prove that $\P(N > n) \to 0$ as $n \to +\infty$. Let $W$ be a random walk with step distribution $1-\xi$ and $I = \inf_{k \geq 1} W(k)$. Then, $W^\uparrow$ is by definition equal in distribution to $W$ under $\P_0(\, \cdot \mid I \geq 1)$ (where the subscript refers to the initial state of $W$). Let moreover $B_n$ be the time at which $W$ advertizes a release for the $n$th time, so that
\[ \P(N = n) = b \P_0 \left( B_n < +\infty, B_{n+1} = +\infty, I \geq 1 \right), \]
with $b = 1/\P_0(I \geq 1)$. Writing the event $\{ I \geq 1 \}$ as the union between the two events $\{ \inf_{1 \leq k \leq B_n} W(k) \geq 1 \}$ and $\{ \inf_{k > B_n} W(k) \geq 1 \}$, the strong Markov property at time $B_n$ entails
\[ \P(N = n) = b \E_0 \left( p(W(B_n)) ; B_n < +\infty, \inf_{1 \leq k \leq B_n} W(k) \geq 1 \right), \]
with $p(w) = \P_w(N = 0, I \geq 1)$. Coupling $W$ under $\P_w$ with a version of $W$ under $\P_0$ that stays below it, it is easy to see that $p(w)$ is increasing in $w$ and so
\begin{align*}
	\P(N = n) & \geq b p(0) \P_0 \left( B_n < +\infty, \inf_{1 \leq k \leq B_n} W(k) \geq 1 \right)\\
	& \geq b p(0) \P_0 \left( B_n < +\infty, I \geq 1 \right).
\end{align*}
This last lower bound is equal to $p(0) \P_0(B_n < +\infty \mid I \geq 1)$ which by definition is equal to $p(0) \P(N > n)$. Since $p(0) = \P(N = 0)/b$ is $>0$ by Lemma~\ref{lemma:>0}, dividing by $p(0)$ leads to
\[ \P(N > n) \leq \frac{b\P(N = n)}{\P(N = 0)}. \]
Since $\P(N = n) \to 0$, this finally proves~\eqref{eq:goal-N} and hence Proposition~\ref{prop:T_1}.

\subsubsection{More on the case $\beta > 1$, $\beta \approx 1$}
\label{subsub:intermediate-beta}

The simulation results in Section~\ref{ec:simula} show a rather fuzzy behavior of $\alpha(\beta)$ for $\beta$ close to~1. Indeed, the curves shown in Figure~\ref{fig:betatoalpha} are smooth for small and large values (say, $\beta < 1/2$ and $\beta > 1.2$) but for $\beta$ close to one it is difficult to obtain stable numerical results. Our goal here is to discuss potential interesting phenomena arising for $\beta > 1$ close to one.

The shape of the function $\rho \mapsto F(\rho, \beta)$ depicted in Figure~\ref{fig:beta2toalpha} is typical for large values of~$\beta$, say $\beta > 1.2$. In particular, this function is increasing which makes the regression of $F$ against $1/\log(1/(1-\rho))$, such as in~\eqref{eq:regression}, reasonable. However, as $\beta$ gets closer to one, the shape of this function changes. For instance, Figure~\ref{fig:beta12toalpha} shows simulation results for $F(\rho, 1.2)$ which are representative of $F(\rho, \beta)$ for small~$\beta$, say $1 < \beta < 1.2$. Noticeably, the function $F(\rho,1.2)$ is not monotone in~$\rho$ and so the approximation~\eqref{eq:regression} cannot be valid for every~$\rho$. Rather, we find that $F(\rho, \beta)$ decreases and then increases, and that the regression against $1/\log(1/(1-\rho))$ is only accurate past the minimum.
\begin{figure}[htbp]
	\vspace{2mm}
    \begin{tikzpicture}%
\begin{axis}[
font=\scriptsize,
width=\columnwidth+1cm, height=5.5cm,
    xmin=2, xmax=8,
    ymin=1.45, ymax=2,
    xtick={2,4,6,8},
    every axis x label/.style=
        {at={(ticklabel cs:0.5)},anchor=center},
    scaled ticks=false,
	xlabel={$\log(1/(1-\rho))$},
    legend cell align=left,
    legend style={at={(1,1)}, anchor=north east, inner xsep=1pt, inner ysep=1pt}]

\legend{
	$F(\rho,1.2)$\\
	Regression result\\%
}
\addplot[color=black] plot coordinates {
	(2.014903,	1.840490)
	(2.420368,	1.794465)
	(2.825833,	1.756408)
	(3.231298,	1.738436)
	(3.636763,	1.728122)
	(4.042229,	1.723681)
	(4.447694,	1.732326)
	(4.853159,	1.755836)
	(5.258624,	1.758578)
	(5.664089,	1.772945)
	(6.069554,	1.793862)
	(6.475019,	1.796923)
}; 
\addplot[color=gray,densely dashed] plot coordinates {
	(2.000000,	1.485550)
	(2.100000,	1.507005)
	(2.200000,	1.526509)
	(2.300000,	1.544317)
	(2.400000,	1.560642)
	(2.500000,	1.575660)
	(2.600000,	1.589523)
	(2.700000,	1.602359)
	(2.800000,	1.614279)
	(2.900000,	1.625376)
	(3.000000,	1.635733)
	(3.100000,	1.645423)
	(3.200000,	1.654506)
	(3.300000,	1.663039)
	(3.400000,	1.671071)
	(3.500000,	1.678643)
	(3.600000,	1.685794)
	(3.700000,	1.692559)
	(3.800000,	1.698968)
	(3.900000,	1.705049)
	(4.000000,	1.710825)
	(4.100000,	1.716320)
	(4.200000,	1.721552)
	(4.300000,	1.726542)
	(4.400000,	1.731305)
	(4.500000,	1.735856)
	(4.600000,	1.740209)
	(4.700000,	1.744377)
	(4.800000,	1.748371)
	(4.900000,	1.752202)
	(5.000000,	1.755880)
	(5.100000,	1.759414)
	(5.200000,	1.762812)
	(5.300000,	1.766081)
	(5.400000,	1.769230)
	(5.500000,	1.772264)
	(5.600000,	1.775189)
	(5.700000,	1.778012)
	(5.800000,	1.780738)
	(5.900000,	1.783371)
	(6.000000,	1.785917)
	(6.100000,	1.788379)
	(6.200000,	1.790761)
	(6.300000,	1.793068)
	(6.400000,	1.795303)
	(6.500000,	1.797469)
	(6.600000,	1.799570)
	(6.700000,	1.801607)
	(6.800000,	1.803585)
	(6.900000,	1.805506)
	(7.000000,	1.807371)
	(7.100000,	1.809185)
	(7.200000,	1.810947)
	(7.300000,	1.812662)
	(7.400000,	1.814330)
	(7.500000,	1.815953)
	(7.600000,	1.817534)
	(7.700000,	1.819074)
	(7.800000,	1.820574)
	(7.900000,	1.822037)
	(8.000000,	1.823463)
}; 
\end{axis}
\end{tikzpicture}%
	\vspace{-7mm}
	\caption{Approximating $\alpha(1.2)$ for $R = 2$.}
	\label{fig:beta12toalpha}
\end{figure}
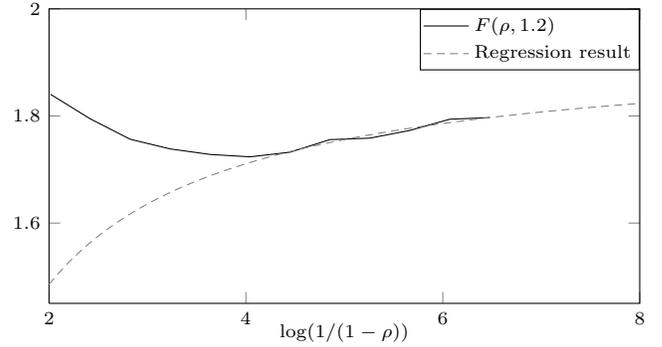

Regressing the curve obtained in Figure~\ref{fig:beta12toalpha} past
the minimum leads to the approximation $\alpha(1.2)\approx 1.94$,
which is still very much in line with our theoretical result
$\alpha(1.2) = 2$. However, the point where the minimum of the function $\rho \mapsto F(\rho, \beta)$ is attained, shifts to the right when $\beta$ gets closer to~1, leaving us with less points against which to regress.
In particular, for $\beta < 1.2$ one would need to simulate the system at loads higher than $0.999$ to get accurate results.

We suspect that this numerical instability is caused by heavy tails phenomena that seem to appear for $\beta < 2$. More precisely, inspecting the proofs of Proposition~\ref{prop:stab} and Theorem~\ref{thm:alpha-beta=infty}, one sees that Corollary~\ref{cor:T^*} is not strong enough for the proofs of the case $\beta = +\infty$ to go through directly. Indeed, instead of controlling the behavior of $\A(T^*)$ and $T^*$ in distribution, one needs (at least, with the proposed proof strategy) to control their mean behavior. Let us do a small computation: let $T_\infty^Z$ be the random variable introduced in Lemma~\ref{lemma:limiting-picture}, which is the weak limit of $T_1 - \tau_1$, and let $B^\uparrow$ be the first time the process $W^\uparrow$ advertizes a release. Then
\[ \E(|T_\infty^Z| ; T_\infty^Z \leq 0) = \E(B^\uparrow; B^\uparrow < +\infty) = \sum_{k \geq 0} k \P(B^\uparrow = k), \]
and, as before, we have
\[ \P(B^\uparrow = k) = \E\left[ \psi(W^\uparrow(k)) \prod_{i < k} (1-\psi(W^\uparrow(i)))\right]. \]
Thus for large $k$, we should have $\P(B^\uparrow = k) \approx k^{-\beta}$ which suggests that $B^\uparrow$, and in particular $|T_\infty^Z|$, has infinite mean for $\beta \leq 2$, although both random
variables are almost surely finite for $\beta > 1$.
It is possible to use this result to show that the (almost surely finite) weak limits of $\A(T^*)$ and $T^* - \tau_{\max}$ have infinite first moment. This potentially invalidates the back-of-the-envelope
computations in Section~\ref{sub:lingering-effect}, and so more care is needed.
For instance, simulation experiments for $\beta = 1.2$ and $R = 2$
suggest that $\E_{a_1}(A_1(T^*))$ grows as $a_1^{0.4}$ as
$a_1 \to +\infty$.

\section{Broader implications} \label{sec:insight}

Motivated by the poor heavy traffic delay performance of ``cautious'' activation rules in queue-based schemes for distributed medium access control, we investigated more aggressive schemes. Our main contribution lies in highlighting a new effect that we called the lingering effect and in studying the performance ramifications of this effect for a special topology. In this section we explain and discuss various directions in which our framework could possibly be extended.

First of all, it would be straightforward to extend our results to the
following asymmetric case: instead of having $R$ queues in each group
with identically distributed arrival processes across queues,
we have two groups of $R_1$ and $R_2$ queues and the arrivals into the
$k$th queue of group $g$ have distribution $\xi_{g,k}$.
We chose to study a symmetric scenario for technical reasons,
since then there is no need to label queues individually.
In this setting, the lingering effect will occur whenever,
informally speaking, the two dominant queues of at least one of the
two groups have the same arrival rate.
For instance, the delay will scale like $1/(1{-}\rho)^2$
with $\rho = \max_k \E(\xi_{1,k}) + \max_k \E(\xi_{2,k})$ if the condition
$\E(\xi_{1,1}) = \E(\xi_{1,2}) \geq \max_{g,k} \E(\xi_{g,k})$ is satisfied.

We believe that the insights provided by the complete bi-partite
constraint graph carry over to more general topologies.
Note that for a general topology, our model needs to be changed,
since one needs then to specify more precisely how queues become active.
One can for instance think of queues going into back-off,
and then trying to grab the channel at some rate.
We conjecture that whenever the constraint graph is not complete
and thus contains an independent set of several nodes, the lingering
effect can rear its head provided some algebraic condition between
the arrival rates at the various queues is satisfied.
It would be very interesting to be able to formulate a precise
and formal conjecture reflecting this intuition,
and most probably even more challenging to prove it.

\bibliographystyle{abbrv}

%
%

\end{document}